\documentclass[pdflatex,sn-mathphys-num]{sn-jnl}

\usepackage{graphicx}%
\usepackage{multirow}%
\usepackage{amsmath,amssymb,amsfonts}%
\usepackage{amsthm}%
\usepackage{mathrsfs}%
\usepackage[title]{appendix}%
\usepackage{xcolor}%
\usepackage{textcomp}%
\usepackage{manyfoot}%
\usepackage{booktabs}%
\usepackage{algorithm}%
\usepackage{algorithmicx}%
\usepackage{algpseudocode}%
\usepackage{listings}%
\usepackage{subcaption}%
\usepackage[export]{adjustbox}%

\theoremstyle{thmstyleone}%
\newtheorem{theorem}{Theorem}
\theoremstyle{thmstyletwo}%

\theoremstyle{thmstylethree}%

\newtheorem{lemma}[theorem]{Lemma}%

\newcommand{{\logiclearner}}{\textsc{LogicLearner}}

\begin{document}

\title[Article Title]{LogicLearner: A Tool for the Guided Practice of Propositional Logic Proofs}

\author*[1]{\fnm{Amogh} \sur{Inamdar}}\email{amogh.inamdar@columbia.edu}

\author[1,3]{\fnm{Uzay} \sur{Macar}}\email{uzay@aiphabet.org}

\author[1]{\fnm{Michel} \sur{Vazirani}}\email{mvv2114@columbia.edu}

\author[2]{\fnm{Michael} \sur{Tarnow}}\email{m.tarnow@columbia.edu}

\author[2]{\fnm{Zarina} \sur{Mustapha}}\email{zarina@columbia.edu}

\author[2]{\fnm{Natalia} \sur{Dittren}}\email{nd2664@columbia.edu}

\author[2]{\fnm{Sam} \sur{Sadeh}}\email{ss6316@columbia.edu}

\author[1]{\fnm{Nakul} \sur{Verma}}\email{verma@cs.columbia.edu}

\author[1,3]{\fnm{Ansaf} \sur{Salleb-Aouissi}}\email{ansaf@cs.columbia.edu}

\affil[1]{\orgdiv{Computer Science}, \orgname{Columbia University}, \country{United States}} 

\affil[2]{\orgdiv{Center for Teaching and Learning}, \orgname{Columbia University}, \country{United States}} 

\affil[3]{\orgdiv{}\orgname{Aiphabet, Inc.}, \country{United States}} 

\abstract{The study of propositional logic---fundamental to the theory of computing---is a cornerstone of the undergraduate computer science curriculum. Learning to solve logical proofs requires repeated guided practice, but undergraduate students often lack access to on-demand tutoring in a judgment-free environment. In this work, we highlight the need for guided practice tools in undergraduate mathematics education and outline the desiderata of an effective practice tool. We accordingly develop {\logiclearner}\footnote{https://logiclearner.ctl.columbia.edu/}, a web application for guided logic proof practice. {\logiclearner} consists of an interface to attempt logic proofs step-by-step and an automated proof solver to generate solutions on the fly, allowing users to request guidance as needed. We pilot {\logiclearner} as a practice tool in two semesters of an undergraduate discrete mathematics course and receive strongly positive feedback for usability and pedagogical value in student surveys. To the best of our knowledge, {\logiclearner} is the only learning tool that provides an end-to-end practice environment for logic proofs with immediate, judgment-free feedback.}



\maketitle

\section{Introduction}\label{sec:intro}

\begin{figure}[t!]
\centering
\includegraphics[width=\textwidth]{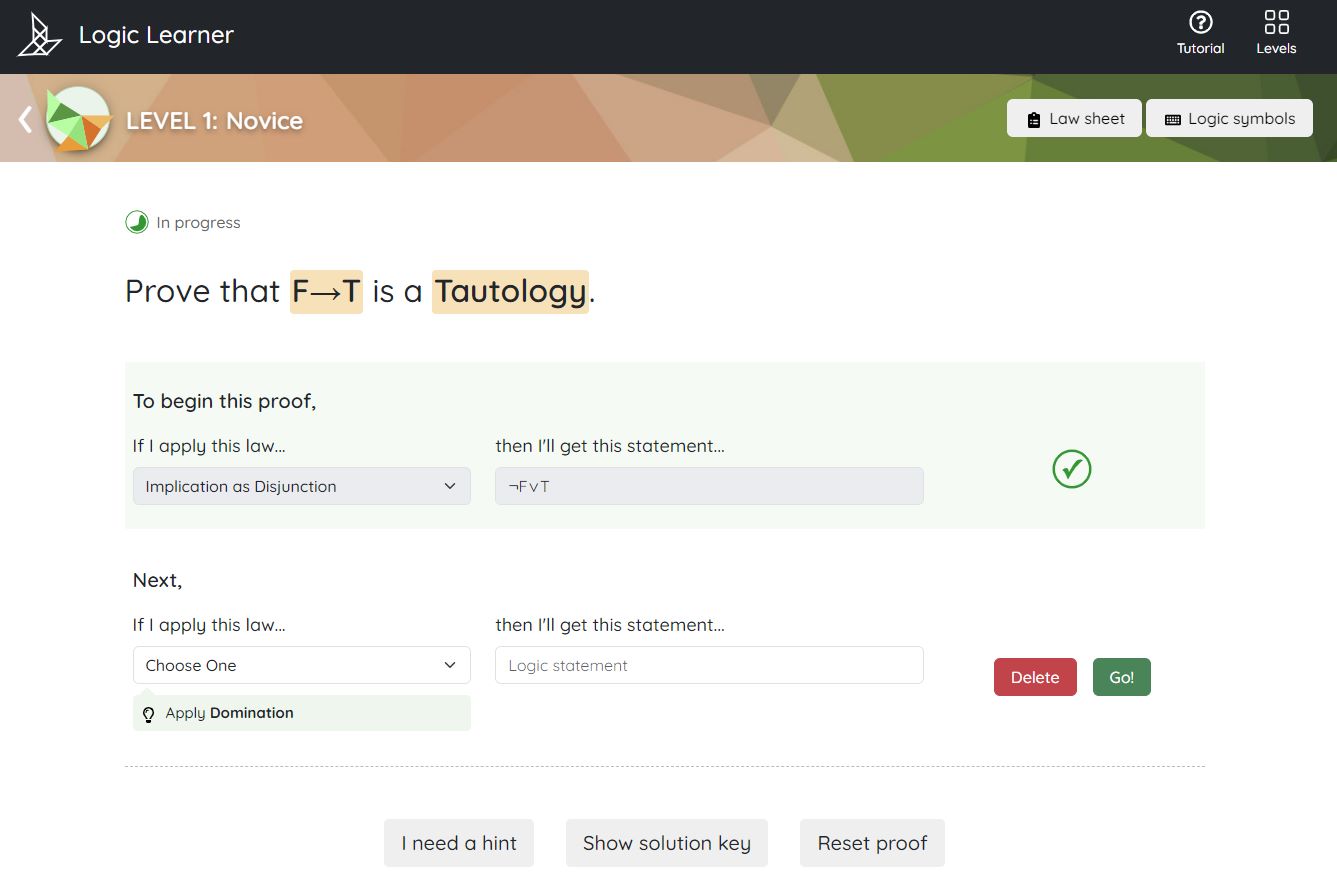}
\caption{{\logiclearner} is a holistic environment for guided logic proof practice.}\label{fig:intro}
\end{figure}

The ability to think critically and reason quantitatively is an important skill for all Computer Science (CS) students. CS students generally acquire such quantitative problem-solving skills in foundational undergraduate proof-based CS courses such as discrete mathematics. However, learning mathematics can be difficult. Students tend to find dealing with proofs challenging, and can seldom construct and communicate a long sequence of logical arguments effortlessly. Students often skip proof steps, leading to logical leaps or incorrect conclusions in their arguments. In order to promote student learning, effective course instructors identify gaps in student understanding and provide timely, high-quality feedback. Student who receive such feedback are able to build confidence, master the material, and develop valuable quantitative reasoning skills \citep{evans2013making}. However, a common issue across educational institutions and grading platforms is that instructors lack the resources to provide timely and informative feedback that promotes student learning. Once the time comes for formal assessments of learning (such as graded homeworks and exams), it is usually too late for the student to master the primary objectives of such courses. This issue is exacerbated in undergraduate-level proof-based courses such as discrete mathematics. There is a wide range in the mathematical maturity of students, and feedback on assignments and quizzes is usually received after several weeks. Additionally, per our exploration of the state-of-the-art, there are currently no off-the-shelf automated grading systems for student-level mathematical proofs to speed up the feedback process. This lack of timely feedback, coupled with a general level of student anxiety in dealing with proofs, translates directly into decreased student performance---our study of exam results in an undergraduate discrete mathematics course found that exam performance was \emph{10\% lower on proof questions} than other types of questions.

Interactive practice in undergraduate-level education is often limited to periodic tutorial sessions with instructors or teaching assistants (TAs). While TAs can do an excellent job in helping students, their limited availability creates a challenge in large undergraduate classes. Additionally, students may fear judgment from peers and instructors over the quality and frequency of their doubts in group settings, leading to reticence in seeking help. To alleviate these shortcomings, we propose the development and use of web-based tools for the guided practice of quantitative reasoning problems. We identify the following desiderata for such practice tools.
\begin{itemize}
    \item \textbf{An effective practice tool is accessible and user-friendly.} Even an excellent practice tool is of little use to students who cannot access it. Tools that are difficult to use or hidden behind paywalls tend to isolate students with limited resources, despite such students being likely to benefit the most from them. 
    \item \textbf{A practice tool must be a source of truth.} As a trusted authority on the subject matter, a teaching aid that provides incorrect information is especially detrimental to the learning process. Ideally, such tools would allow external experts to independently verify the correctness of these tools. We encourage open-source development as a strategy to ensure that pedagogical tools are both accessible and transparent.
    \item \textbf{An effective practice tool is non-judgmental.} While expert feedback is invaluable, a student who fears judgement from peers and instructors will likely be more willing to explore ideas in a self-guided setting. This is why a good practice tool should provide constructive feedback and encouragement to students. Overly critical feedback or high-stakes settings may discourage use and affect the learning process.
    \item \textbf{An effective practice tool provides only relevant guidance.} Effective instructors know that feedback must be catered to the problem setting and the type of mistake made. A tool that inundates the student with options or provides very basic feedback (e.g., only grading an answer as right or wrong) is likely to be less effective than one with specific, cogent, and relevant guidance.
\end{itemize}

With these desiderata in mind, we tackle the challenge of improving the pedagogical experience of learning to solve logic proofs. We develop {\logiclearner} (Figure \ref{fig:intro}), an open-source web application that provides a holistic environment for logic proof practice. {\logiclearner} minimizes the feedback turnaround time for logic proof practice, leading to greater student engagement with course materials related to theorems and proofs. Using tools from AI and machine learning, {\logiclearner} validates student solutions at each steps, identifies mistakes, and solves the problem on-the-fly to provide personalized hints in a judgment-free environment. {\logiclearner} is intended to supplement, not replace, the invaluable efforts of teaching assistants and tutorial sessions in a course that covers propositional logic. {\logiclearner} is a free and open-source project, managed by the Center for Teaching and Learning (CTL) at Columbia University. Visuals of the application are presented in Appendix \ref{sec:apxVisual}. The source code is available at \url{https://github.com/ccnmtl/logiclearnertools}. 
\\

Our contributions in this work can be summarized as follows:
\begin{itemize}
    \item We identify a gap in the availability of on-demand, judgment-free guidance in undergraduate computer science curricula. 
    \item We highlight the need for practice tools that provide instantaneous, high-quality feedback, and identify the desiderata of these tools for pedagogical impact.
    \item We develop {\logiclearner}, a web-based application for logic proof practice with instant guidance that is mathematically sound and judgment free. At the time of writing, {\logiclearner} is the only application for the end-to-end practice of logic proofs with guidance.
    \item Through user studies across two semesters of an undergraduate discrete mathematics course, we show that {\logiclearner} reduces student anxiety and improves their confidence in tackling logic proofs.
    \item We assess the drawbacks of AI question-answering tools (a popular alternative among students for pedagogical guidance) and show that {\logiclearner} provides superior pedagogical value over such tools in multiple ways.
\end{itemize}

\section{Related Work}\label{sec:related}

Our work combines the advances in automated problem solving and in software for mathematics/logic education. In this section, we analyze the state of the art in these areas though a pedagogical lens.

\subsection{Automated logic problem solvers}\label{subsec:rel1}

To compensate for a lack of on-demand, judgment-free pedagogical guidance, students are increasingly turning to AI question-answering systems powered by Large Language Models (LLMs), like ChatGPT\footnote{OpenAI (2023). ChatGPT (Mar 14 version); GPT 3.5 backend. \url{https://openai.com/blog/chatgpt}}, for pedagogical feedback. LLMs \cite{achiam2023gpt, touvron2023llama, vicuna2023} are neural networks with billions of parameters that are trained for text generation on internet-scale datasets. With previously unseen prowess on natural language tasks, LLMs are now the basis of applications across domains ranging from medicine to finance \cite{thirunavukarasu2023large, singhal2022large, cui2023chatlaw, webersinke2022climatebert, wu2023bloomberggpt}. However, using an LLM-based application as a learning tool for mathematics has several drawbacks. LLMs require the use of specialized prompting techniques to score well on mathematical reasoning tasks. Chain-of-Thought prompting \cite{wei2022chain} formulates step-by-step LLM prompts to emulate human reasoning, improving on question-only prompts for simple arithmetic and logic problems. Extensions such as Graph of Thoughts \cite{besta2023graph} and MathPrompter \cite{imani2023mathprompter} develop increasingly complex input- and early-layer techniques for sophisticated users. Drori et al. \cite{drori2022neural} leverage the OpenAI Codex LLM \cite{chen2021evaluating} to produce impressive question-answering performance on undergraduate mathematics courses, but explicitly state that the model cannot solve ``questions with solutions that require proofs''. Datasets that have been developed to evaluate reasoning ability \citep{ontanon2022logicinference, cobbe2021training, hendrycks2021measuring} primarily consist of problems with simple, few-step solutions. Additionally, LLMs have been shown to hallucinate factual information \cite{huang2023survey} and are susceptible to prompt manipulation \cite{cohen2024comes}. Hence, students relying on LLMs for assistance with proofs are at risk of receiving \emph{confident but incorrect} answers that stunt the learning process.

Another focus in the literature is on tool-assisted human problem solving. Automated theorem-provers such as Isabelle \cite{nipkow2002isabelle} have been used to develop human-friendly representations of logic problems \cite{Villadsen_2022} and other mathematical structures \cite{fuenmayor2022formalising}. Other tools---such as the popular mathematical software Wolfram \cite{weisstein}---simplify or evaluate Boolean expressions using truth tables, but are not directly useful for the development of proof-solving skill.

While these works make impressive strides in AI mathematical reasoning, they are not yet advanced enough to solve problems requiring several steps of reasoning. Most importantly, these systems \emph{are not designed for pedagogical use}. They do not provide interfaces that benefit problem-solving practice and produce full answers, preventing students from reasoning by themselves with only subtle hints when needed. Using such systems as practice tools compounds students' frustrations with the learning process. To address this important gap and facilitate practice, we develop {\logiclearner} as purpose-built learning tool that is both mathematically sound and easy to use. 

\subsection{The pedagogy of propositional logic}\label{subsec:rel2}

Propositional logic (and discrete mathematics at large) are important topics of study---not only for computer science students, but also for general mathematical maturity \citep{sandefur2022teaching, greefrath2022mathematical}.
The literature largely asserts that these are challenging topics to learn. Multiple prior works have attempted to model student behavior in learning to solve proofs. Dawkins and Roh \cite{dawkins2022aspects} posit that students have a twofold process to understanding a logical property, first generating examples and then evaluating which examples satisfy the property. This generate-and-test process is analogous to many fundamental computer search algorithms. EvoLogic \cite{galafassi2020evologic, galafassi2022evologic} takes an agent-based approach to model student responses to 10 simple logic exercises. These works do not attempt to improve the learning process, as we do in this paper. 

Logical concepts also form the basis of a few learning games. \textit{Proplog} and \textit{Syllog} \cite{ohrstrom2019teaching} require players to evaluate the logical soundness of propositions and syllogisms respectively. TrueBiters \cite{de2019truebiters} aims to gamify the process of learning truth tables by representing bitwise operators as monsters that `eat' bits and return the results, with the goal of reducing a bit string to a specified target bit. We are not aware of any games or practice tools that directly aim to build proof-solving skills, as {\logiclearner} does.

\section{Methods}\label{sec:method}

{\logiclearner} is a web application with a gamified interface that is purpose-built for learning, backed by a parser to process Boolean expressions and an AI proof solver to provide hints when requested. In this section, we describe the working of each of its components in detail. Starting with the user interface, each subsequent sub-section describes a module that is increasingly abstracted away from the student. We reiterate that {\logiclearner} is open-source and encourage community contributions to its improvement.

\subsection{User Experience}\label{subsec:m1}

As described in Section \ref{sec:intro}, ease of use is essential to a good pedagogical tool. We design {\logiclearner} as a web application with an attractive but simple interface that allows students to focus on solving the problems. For first-time users, {\logiclearner} contains a tutorial section with 6 simple steps describing usability. Students then attempt logic problems of varying difficulty. As they go through the learning process, {\logiclearner} tracks their progress and enables them to review and re-attempt previously attempted proofs. As a web application, {\logiclearner} \textbf{works well on both computers and mobile devices}, allowing for practice on the go.

\subsubsection{User Interface}\label{subsubsec:m11}

\begin{figure}[ht!]
    \centering
    \includegraphics[width=0.8\linewidth]{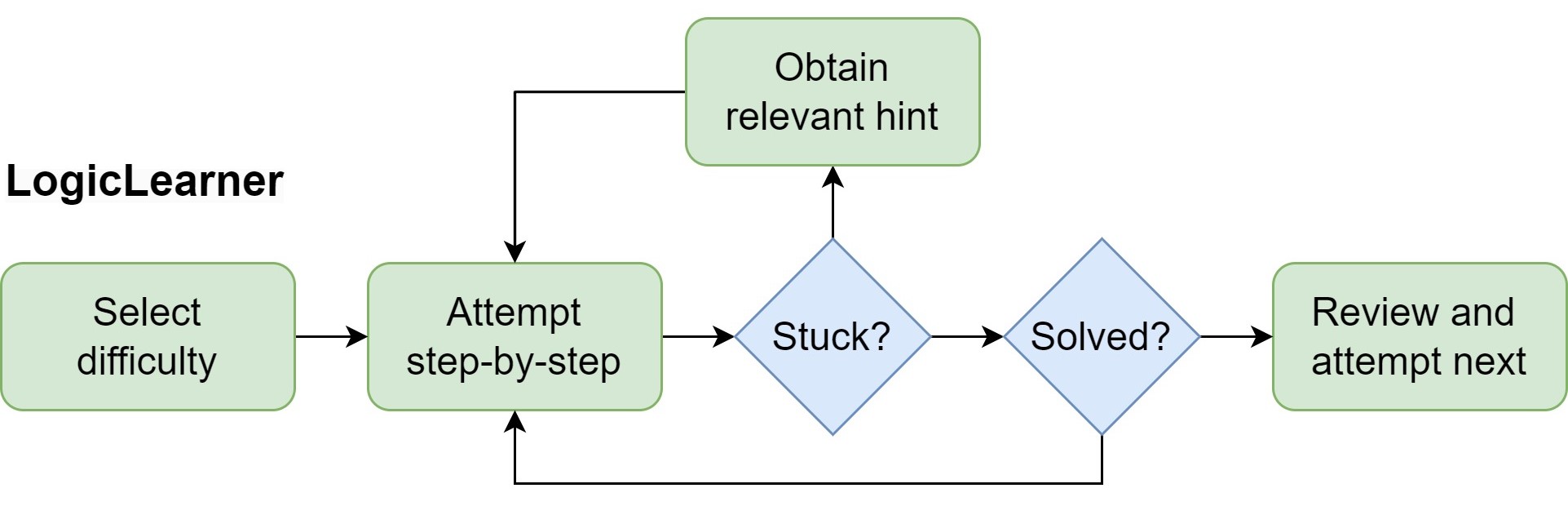}
    \caption{The user flow diagram of {\logiclearner} shows that students have a linear decision making process, enabling ease of use.}\label{fig:uflow}
\end{figure}

As seen in Figure \ref{fig:uflow}, students can choose the difficulty of the proof they will attempt. As a student attempts a proof, {\logiclearner} finds solutions from the steps they have already taken and produces hints when requested to unblock their progress. Once a solution has been found (or requested), the student is able to review the full proof and attempt the next question if desired. Detailed screenshots of the application, including the tutorial section, can be found in Appendix \ref{sec:apxVisual}.

Students use buttons to navigate the application. At each proof step, they select a logic rule to apply via a drop-down menu and enter the corresponding statement as free-form text. If incorrect, students are not explained their mistakes to allow them to reason again. Upon requesting a hint, students are provided with the correct logic rule to apply. Another hint request leads to the correct expression for the next step being provided. We use this two-step hints process to allow students to attempt problems with only partial information, just as an instructor would nudge a student in the right direction. To prevent student frustration when stuck, we also allow them to view the full solution without penalty if requested. Students are able to reset their progress and re-attempt the questions \emph{ad infinitum}. The user interface of {\logiclearner} is implemented as a Django\footnote{\url{https://www.djangoproject.com/}} web server in Python.

\subsubsection{Database}\label{subsubsec:m12}

{\logiclearner} maintains a PostgreSQL\footnote{\url{https://www.postgresql.org/}} database as the back end of the user experience. Students do not need to explicitly log in to the application, but their progress is still preserved across sessions by tracking and storing site accesses. This provides a seamless user experience where students are not prompted to repeat proofs they have already attempted unless explicitly requested. {\logiclearner} is a fully open source application, and source code is available in public repositories on GitHub\footnote{UX: \url{https://github.com/ccnmtl/logiclearner}, Back-end: \url{https://github.com/ccnmtl/logiclearnertools}}. The application server and database are actively managed by Center for Teaching and Learning at Columbia University.

\subsection{Application Back-end and Proof Solving}\label{subsec:m2}

{\logiclearner} performs several computations upon receiving user input. When a student provides a logic rule and an expression as the next step of a proof, the application must 
\begin{itemize}
    \item validate the syntax of the input expression,
    \item check for logical entailment from the current state, and
    \item find the subsequent steps of a full solution for hints.
\end{itemize}

This process---described as a data flow diagram in Figure \ref{fig:methods}---requires a computational model of Boolean expressions and logical proofs. We model Boolean expressions with a Context-Free Grammar (CFG) and describe the logic proof as a graph search problem between expressions. We then use AI search techniques to solve these proofs and provide hints to students in real time. Below, we describe this `business logic' of {\logiclearner} in detail.

\begin{figure*}[ht!]
    \centering
    \includegraphics[width=\linewidth]{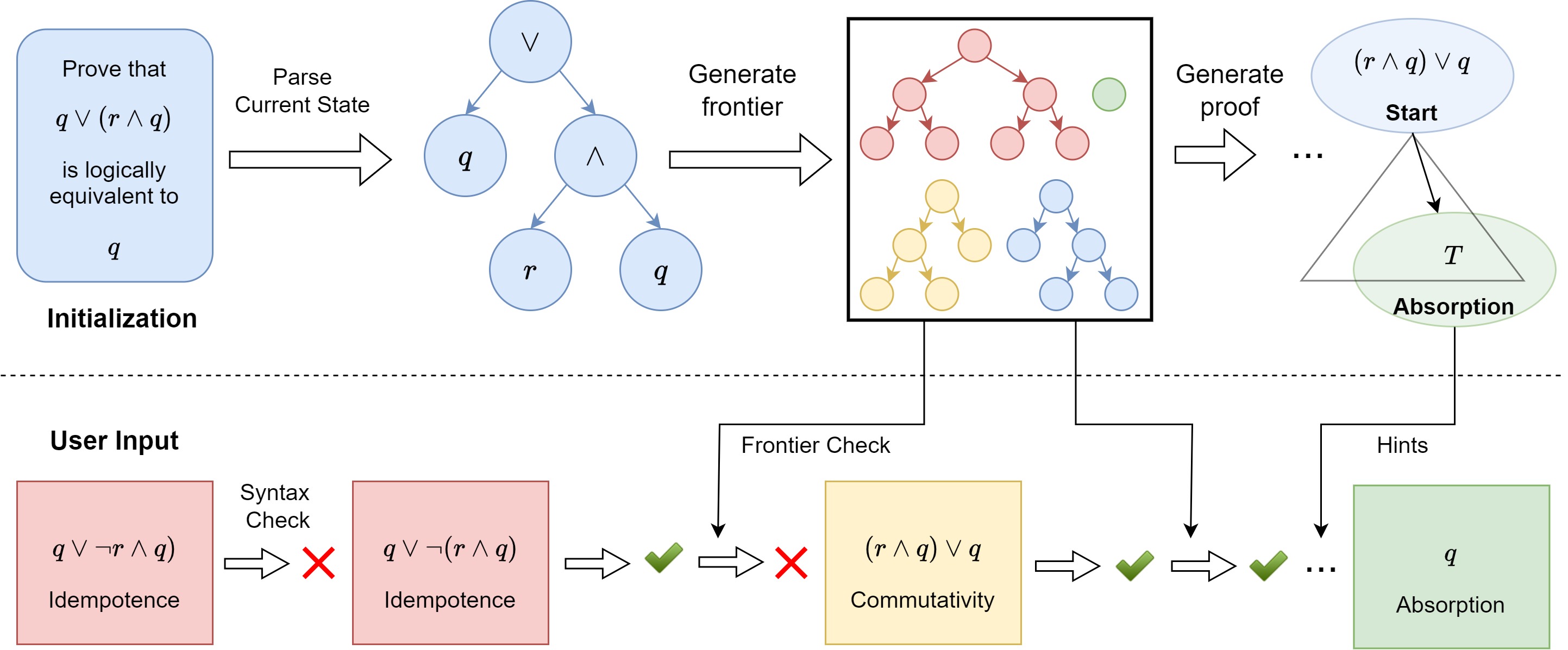}
    \caption{{\logiclearner} generates a `frontier' of all possible next steps from the current expression, and subsequently a full proof for hints. User input is checked for syntax and membership in the frontier. The user receives feedback on incorrect inputs and can request a hint at any time.}\label{fig:methods}
\end{figure*}

We use the following notation in this and subsequent sections: $X = \{x, y, z,\ldots\}$ is a set of Boolean variables. $E$ denotes the set containing every Boolean expression $e$ formed by any number of variables from $X$ and the constants $\{T, F\}$, which denote \textbf{True} and \textbf{False} respectively. $R$ denotes the set of logic rules, which are transformations according to the named logical equivalences between expressions (Idempotence, Associativity, etc., Table \ref{tab:rules}). Every $r\in R$ transforms expression $e$ into one of a set of expressions $\{e_1, e_2,\ldots\}$, where each $e_i$ is formed by applying $r$ to a different position in $e$. If $r$ cannot be applied to $e$ then $r(e)=\emptyset$. To fully represent an equivalence relation, we define the inverse of a rule $r$ to be $r^{-1}\in R$ such that $e\in r^{-1}\circ r(e)$. For example, we can apply Idempotence ($r$) to the expression $p\lor q \lor q$ to derive $p\lor q$, and likewise apply Idempotence ($r^{-1}$) to $p\lor q$ to derive $p\lor q \lor q$. We phrase our problems as `\textit{Prove that \textbf{source} is logically equivalent to \textbf{target}.}', where \textbf{source} and \textbf{target} are Boolean expressions. When \textbf{target} is \textbf{True}/\textbf{False}, the question is phrased `\textit{---is a Tautology}'/`\textit{---is a Fallacy}' respectively. Every problem is guaranteed to have a solution.

\begin{table}[ht!]
\begin{center}
\begin{minipage}{\textwidth}
\caption{The rules of propositional logic}\label{tab:rules}%
\begin{tabular}{@{}lll@{}}
\toprule
Rule & Equivalence 1 & Equivalence 2  \\
\midrule
Absorption      &  $p \lor (p\land q) \equiv p$ & $p \land (p\lor q) \equiv p$ \\
Associativity   &  $p \lor (q\lor r) \equiv (p \lor q) \lor r$ & $p \land (q\land r) \equiv (p \land q) \land r$ \\
Commutativity   & $p \lor q \equiv q \lor p$ & $p \land q \equiv q \land p$ \\
De Morgan's Law & $\neg(p \lor q) \equiv \neg p\land \neg q$ & $\neg(p \land q) \equiv \neg p\lor \neg q$ \\
Distributivity  & $p \lor (q\land r) \equiv (p \lor q) \land (p \lor r)$ & $p \land (q\lor r) \equiv (p \land q) \lor (p \land r)$ \\
 & $p \lor (q\lor r) \equiv (p \lor q) \lor (p \lor r)$ & $p \land (q\land r) \equiv (p \land q) \land (p \land r)$ \\
Domination      &  $p \lor T \equiv T$ & $p \land F \equiv F$ \\
Idempotence     &  $p \lor p \equiv p$ & $p \land p \equiv p$ \\
Identity        & $p \lor F \equiv p$ & $p \land T \equiv p$  \\
Iff as Implication         & $p \leftrightarrow q \equiv (p \to q) \land (q \to p)$ & -- \\
Implication as Disjunction &  $p \to q \equiv \neg p \lor q$ & -- \\
Negation        & $p \lor \neg p \equiv T$ & $p \land \neg p \equiv F$ \\
\botrule
\end{tabular}
\end{minipage}
\end{center}
\end{table}

\subsubsection{Validating user input}\label{subsubsec:m21}

To take a step towards the solution from their current state $e_t$, a user selects a logic rule $r$ and inputs a Boolean expression $e_{t+1}$. The input is `valid' if it is a syntactically correct Boolean expression that is entailed by the selected rule of logic, i.e., $e_{t+1}\in r(e_t)$. 

We use a left-recursive context-free grammar for Boolean expressions and use the Lark\footnote{\url{https://pypi.org/project/lark-parser}} parser-generator to validate user syntax. The grammar accounts for variations in token representation (such as `1', `T', or `True' for the \textit{true} truth value) and requires a maximum look-ahead of 1, enabling the use of Lark's fast and efficient LALR(1) parser. Each parse tree node is annotated the start and end positions of its token span.

\subsubsection{Modeling the logic proof}\label{subsubsec:m23}

Solving logic problems is challenging to students who first encounter them. To guide students who are stuck, we develop a two-level hints feature that can be triggered to provide the correct rule and subsequently the correct expression for the next step. To do this, {\logiclearner} must find a path to the target expression $e_t$ from the user's current expression $e_c$. The size of the Boolean state space and unpredictable user input make pre-computed or brute-force solutions infeasible. Here, we model logic proofs as a graph search problem and use A* search \cite{Hart1968} to dynamically find solutions. We define some properties and assumptions of our search graph below: 
\begin{enumerate}[1.]
\item Every node represents a Boolean expression $e$. For a given proof, the student must derive expression $e_t$ starting from expression $e_s$.
\item Every edge represents a transition between nodes brought about by applying a logic rule $r_i$. That is, $\exists$ edge $(e_p,e_q)$ $\forall i, p, q$ s.t. $r_i(e_p) \supseteq e_q$.
\item Since every rule $r_i\in R$ is an equivalence relation, every edge $(e_p, e_q)$ is bidirectional.
\item We only focus on `solvable' problems of the type \textit{`Prove that $e_s$ is logically equivalent to $e_t$'}, so that at least one path (proof) from $e_s$ to $e_t$ is known to exist.
\end{enumerate}

These properties ensure that every attempt of a {\logiclearner} problem always has a path to success.

\begin{lemma}\label{lemma:1}
 In every {\logiclearner} problem, the target expression $e_t$ is always reachable from any expression $e_c$ that a student derives from the premise $e_s$.  
\end{lemma}
\begin{proof}
Suppose the student applied logic rules $r_1, r_2,\ldots,r_n$ in sequence to start expression $e_s$ to reach $e_c$ (i.e., $e_c\in r_n\circ r_{n-1}\circ\cdots\circ r_1(e_s)$). Property 3 lets us apply their corresponding inverses to $e_c$ to obtain $e_s$ (i.e., $e_s\in r^{-1}_1\circ r^{-1}_2\circ\cdots\circ r^{-1}_n(e_c)$). By Property 4, $e_t$ is reachable from $e_s$.
\end{proof}

Thus, solving a logic proof is equivalent to searching a connected graph containing $e_s$ and $e_t$. Since rules like Idempotence and Absorption can be chained indefinitely, this graph has an infinite depth. Note that the method in our proof may not produce the shortest path from $e_c$ to $e_s$. Our AI search strategy is also motivated by the fact that hints that backtrack to a start state make for a poor user experience.

\subsubsection{Finding a search frontier}\label{subsubsec:m22}

The search frontier $F(e)$ of a Boolean expression $e$ is the set of all expressions obtained applying a logical equivalence to $e$, i.e., $F(e) = \{e' | \exists r\in R: e'\in r(e)\}$. \textbf{FRONTIER\_GEN}, our efficient frontier-generation algorithm, uses Lark's \textbf{Transformer} feature to traverse up the expression's annotated parse tree from leaves to root, replacing each node by the set of possible transforms at that node. This culminates at the root as a set of all potential next-steps via logical substitution. \textbf{FRONTIER\_GEN} visits each subtree of the parse tree exactly once and applies each logic rule to the token span at that root, keeping all valid transformations.

A sub-expression of an expression $e$ is a contiguous span of $e$'s tokens that forms a syntactically correct expression, appearing as a sub-tree in $e$'s parse tree. The smallest sub-expression, a single literal (Boolean variable or constant), forms a leaf node in the parse tree. Thus, an expression $e$ containing $|e|$ literals ($\le$ num. leaves) has at most $2|e|$ sub-expressions ($\le$ num. subtrees). The number of logic rules $|R|$ is constant and we implement each rule in $\mathcal{O}(1)$ operations. Hence, the size of the search frontier $F(e)$ is $\mathcal{O}(|e|)$ and \textbf{FRONTIER\_GEN} finds this frontier in $\mathcal{O}(|R|\cdot|E|)=\mathcal{O}(|E|)$ operations, which is worst-case optimal.

\begin{algorithm}[hbt!]
\caption{FRONTIER\_GEN}\label{alg:cap}
\begin{algorithmic}
\Require Annotated parse tree $T(e)$ of Boolean expression $e$, Rule set $R$
\Ensure Search frontier $F(e)$
\State $F \gets \emptyset$ 
\While{Transform(e)}  \Comment{Lark Transformer iterates leaf-to-root}
    \State $n \gets \verb|get_current_node|()$
    \State $t_s,t_e \gets \verb|get_token_span|(n)$
    \State $S_n \gets \verb|apply_logic_rules|(R, n)$  \Comment{Set of next-step sub-expressions}
    \While{$S_n \neq \emptyset$}
        \State $s \gets \verb|get_element|(S_n)$
        \State $F \gets F\bigcup \verb|concat|(e[:t_s],s,e[t_e:])$  \Comment{Sub. rule result into token span}
        \State $S_n \gets S_n - \{s\}$
    \EndWhile
\EndWhile
\end{algorithmic}
\end{algorithm}

\subsubsection{Solving proofs with A* Search}\label{subsubsec:m24}

The A* graph search algorithm aims to find the lowest path between two nodes on a weighted graph using heuristic approximations of unknown path costs. A `consistent' heuristic---one that obeys the Triangle inequality---guarantees that A* search finds an optimal path between two nodes \citep{Hart1968}. This makes it well-suited to solving logic proofs when structured as a Boolean graph search problem. However, we are not aware of an efficient heuristic that consistently estimates the semantic similarity (shortest path) between arbitrary Boolean expressions. Instead, we focus our efforts on a variety of heuristics that approximate the path length between expressions. 

\begin{figure*}[ht!]
    \centering
    \begin{subfigure}[t]{0.8\textwidth}
        \centering
        \includegraphics[width=\linewidth]{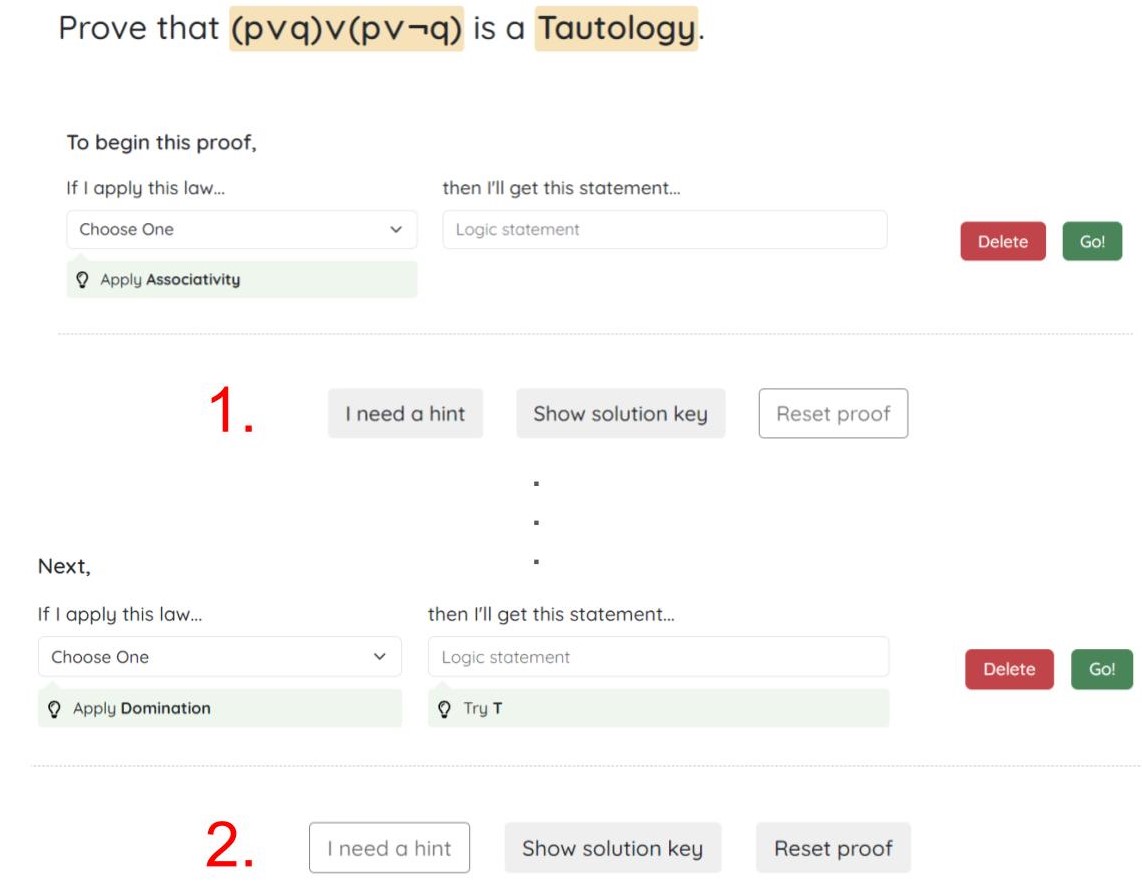}
        \caption{}
    \end{subfigure}%
    
    \begin{subfigure}[t]{0.9\textwidth}
        \centering
        \includegraphics[width=\linewidth]{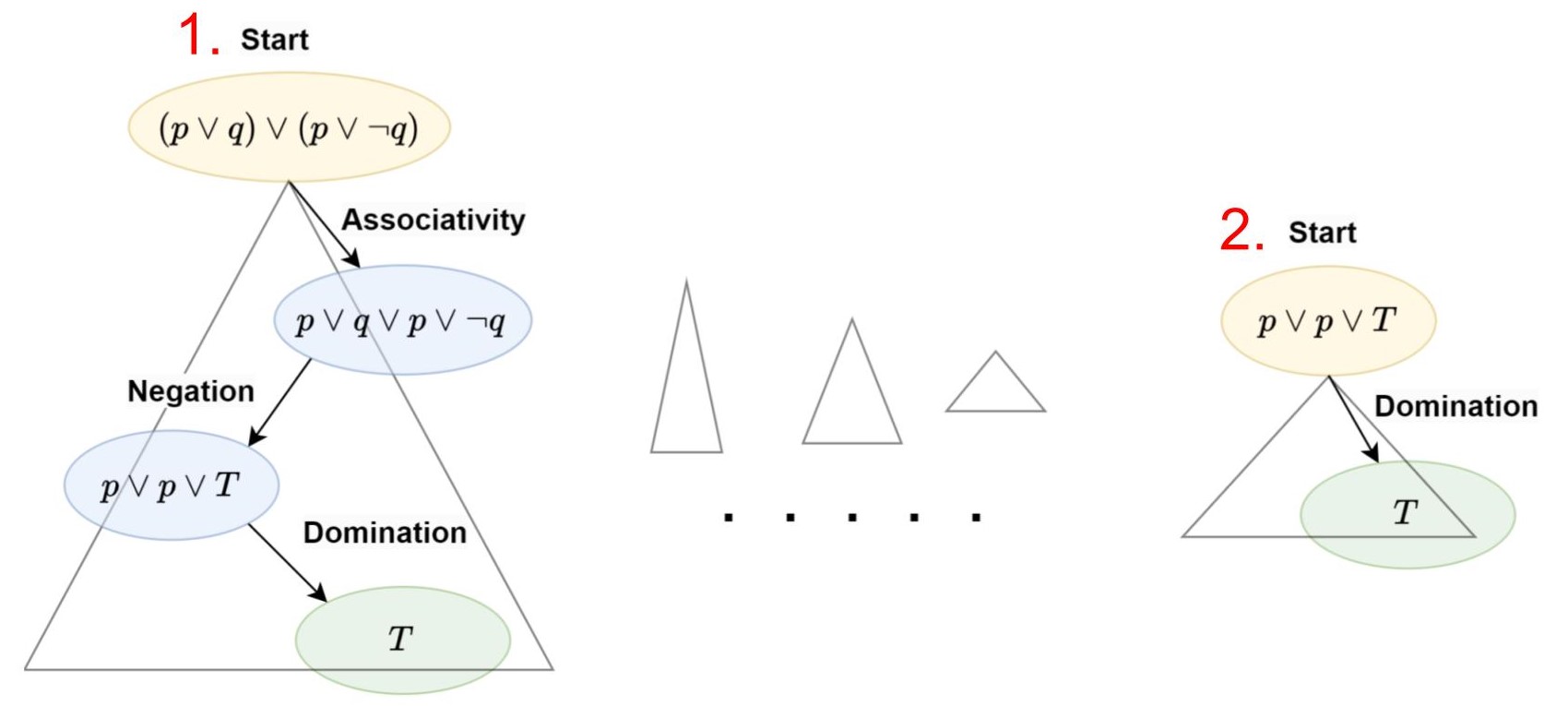}
        \caption{}
    \end{subfigure}
    \caption{(a) A users attempts a logic proof on {\logiclearner} and requests hints when stuck. (b) Search trees are calculated at every step of the proof attempt. The trees for the two requested hints in part (a) are as depicted here as numbered.}\label{fig:proofTrees}
\end{figure*}

As a first approximation, we use a weighted linear combination of measures of similarity between string representations of Boolean expressions (Table \ref{tab:heur}). Since logic rules are not applied with uniform probability, we also consider the rule used to generate a search node. Since computing a gradient over our objective (number of questions solved) is infeasible, we optimize this combination with a genetic algorithm \cite{holland1973genetic} which `evolves' a population of weighted combinations of heuristics (candidate solutions). Candidate fitness is evaluated on the training problem set and the next generation is chosen via fitness-proportionate selection with elitism. The weight ranges, crossover and mutation probabilities, and degree of elitism are determined empirically. Our genetic algorithm produces a marked improvement in performance compared to a randomly weighted ensemble, but we do not argue that the resulting heuristic is consistent.

In production, we implement a time-bound, depth-limited A* path-finding algorithm that searches the Boolean graph space for the target node beginning from the start node. We use the above ensemble, optimized on our question bank, as our heuristic of choice. Since {\logiclearner} provides real-time hints, we constrain its search time to a few seconds and return the path to the lowest-cost node if the solution isn't found. Limiting the search depth compensates for inconsistent heuristics by preventing meandering search paths. As an illustration, stylized search trees corresponding to hints requested on the interface are shown in Figure \ref{fig:proofTrees}. Details about our training methodology, hyperparameter search, and production heuristic can be found in Appendix \ref{sec:apxAblate}.

\begin{table}[ht!]
\begin{center}
\begin{minipage}{\textwidth}
\caption{A* search heuristics}\label{tab:heur}%
\begin{tabular}{@{}ll@{}}
\toprule
Heuristic & Description (computed on a transform $(e_1, r, e_2)$ where $e_2\in r(e_1)$) \\
\midrule
Unitary function & Returns 1 regardless of input \\
Levenshtein distance & The number of single-character edits required to transform $e_1$ to $e_2$ \\
Variable mismatch & The number of variables that appear either only in $e_1$ or only in $e_2$ \\
Length difference & Absolute difference in string length between $e_1$ and $e_2$ \\
Rule weight (for each $r\in R$) & A prior on rule $r$ proportional to its frequency in training data \\
\botrule
\end{tabular}
\end{minipage}
\end{center}
\end{table}

\subsection{Extensions: question generation and neural embeddings}\label{subsec:m3}

While our production heuristic is efficient and effective, it approximates the complex semantics of Boolean algebra with only a few surface-level comparisons. A natural extension to this is to integrate \emph{semantic} information into our heuristic ensemble by modeling the space of Boolean expressions in a way that allows for efficient approximations of similarity. As a proof-of-concept, we develop a neural network to learn the semantics of Boolean expressions for proof-solving though auxiliary tasks.

Neural networks are powerful machine learning models that are able learn complex relationships present in large datasets. Since these models require more data than can be manually curated, we develop \textbf{PROOF\_GEN} (Algorithm \ref{alg:qgen}) to automatically generate logic proofs. Starting from a given target expression $e_t$, \textbf{PROOF\_GEN} generates a search frontier and randomly selects an expression $e_{t-1}$ from it. This procedure is repeated to obtain $e_{t-2},\ldots,e_{t-k+1}=e_s$. Reversed, this is now a $k$-step proof $e_s\rightarrow\cdots\rightarrow e_t$. This also provides us with a way to expand {\logiclearner}'s question bank---\textbf{PROOF\_GEN} does not guarantee elegance, but cherry-picking interesting proofs from its output is an easier task than manually composing novel proofs.

\begin{algorithm}[hbt!]
\caption{PROOF\_GEN}\label{alg:qgen}
\begin{algorithmic}
\Require Boolean expression $e$, Number of proof steps $N$
\Ensure Logic proof $P(e',e)$
\State $P \gets [\ ]$ 
\State $i \gets 0$
\While{$i < N$}
    \State $F \gets \verb|FRONTIER_GEN|(e)$
    \State $e' \gets \verb|random_select|(F)$
    \State $P \gets \verb|append|(P, e')$
\EndWhile
\State $P \gets \verb|reverse|(P)$
\end{algorithmic}
\end{algorithm}

We use \textbf{PROOF\_GEN} generate large amounts of data for learning Boolean semantics, and train encoder-decoder neural networks on two tasks: rule prediction and proof length prediction. To evaluate the potential of neural networks to embed Boolean expressions in metric space while preserving semantic similarity, we use the cosine similarity between encoded Boolean expressions as a heuristic to A* search on our human-curated question bank. We also evaluate an untrained baseline and a language model pre-trained on a multilingual text corpus. Details of these experiments and their results are shown in Appendix \ref{sec:apxNN}. \textbf{PROOF\_GEN} does not have graphical interface at the time of writing, but can be accessed via the \verb|logictools| library API.

\section{Results}\label{sec:result}

To evaluate its effectiveness, we piloted {\logiclearner} as an optional practice tool over two semesters of the undergraduate discrete mathematics class (COMS 3203) at Columbia University. We surveyed students on their confidence in solving logic proofs before and after covering the propositional logic unit of the course. We also recorded their assessment of {\logiclearner}'s features after learning. We present these results in Section \ref{subsec:r1}. In Section \ref{subsec:r2}, we compare the performance of {\logiclearner} with ChatGPT, an online question-answering tool powered by a Large Language Model (refer to Section \ref{subsec:rel1} for an overview). Analyzing ChatGPT's most common failure modes, we show that LLMs are not currently suited for proof solving and pedagogy. We also present technical results on the performance of {\logiclearner}'s AI proof solver and justify our choice of A* heuristic.

\subsection{User Study: {\logiclearner} as a practice tool}\label{subsec:r1}

To evaluate {\logiclearner} as a tool for guided practice, we conducted a user study and analyzed feedback from the students of COMS 3203, the undergraduate discrete mathematics course at Columbia University, across two semesters. Before and after the propositional logic unit, we surveyed students on their confidence in solving logic proofs in various scenarios. Confidence was self-assessed along the 7-point Likert scale \cite{likert1932technique}, ranging from `Not at all confident' to `Completely confident'. Scenarios ranged from \emph{writing a proof with assistance} to \emph{writing a proof in an exam setting}. As seen in Figure \ref{fig:aconf}, student confidence in \emph{writing a logical proof completely and correctly} improves dramatically after lessons in propositional logic and practice with {\logiclearner}. In particular, students gained confidence in solving proofs in an exam setting, the hardest scenario in our survey. This was also reflected in the course outcomes---adjusted for overall difficulty, we observed that the mean student performance on logic proof questions in the COMS 3203 exams \textbf{increased by roughly 5\%} after the introduction of LogicLearner as an optional practice tool. This is a notable improvement as the class sizes are large (300+) and overall performance on these exams is around 80\% on average, making it unlikely that large changes in results can be observed. 

\begin{figure}[ht!]
\includegraphics[width=\textwidth]{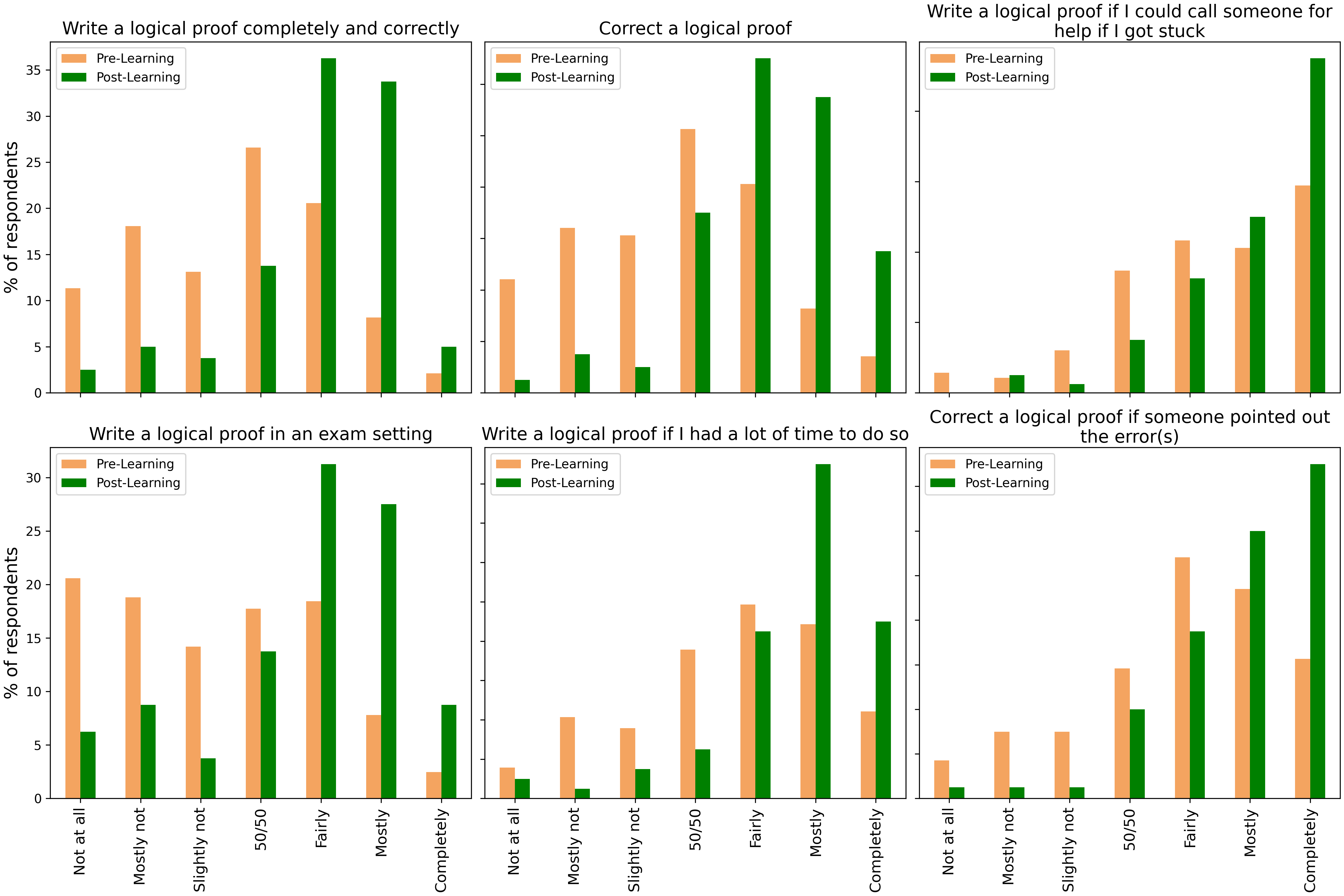}
\caption{Confidence scores for solving logic problems under various scenarios, aggregated across two semesters of a discrete mathematics course}\label{fig:aconf}
\end{figure}

Post-learning, students were prompted to rate the user-friendliness of {\logiclearner} on a scale of 1 to 5 and answer long-form questions on the application's performance. Figure \ref{fig:afeat} shows students rate {\logiclearner}'s highly, with every feature scoring a mean of approximately 4 out of 5 with a standard deviation of less than 1. We analyze the student responses to long-form questions for positive/negative intent using a sentiment classification pipeline from Hugging Face\footnote{\url{https://huggingface.co/}} (Figure \ref{fig:asentiment}). Two of the questions are negative leading questions (`Missing features/improvements' and `Ease of use/challenges in usage'), and negative responses are expected. Negative responses mostly comprise of suggestions for improvement and positive responses indicate where students found nothing lacking. The two open-ended questions (`(Would you) Recommend to others' and `(Would you) Use it to practice') show overwhelmingly positive responses, showing that respondents found {\logiclearner} to be a useful tool for logic practice. The phrasing of the long-form questions and full respondent statistics can be found in Appendix \ref{sec:apxSurvey}.

\begin{figure}[ht!]
\centering
\begin{minipage}{.5\textwidth}
  \centering
  \includegraphics[width=0.9\linewidth]{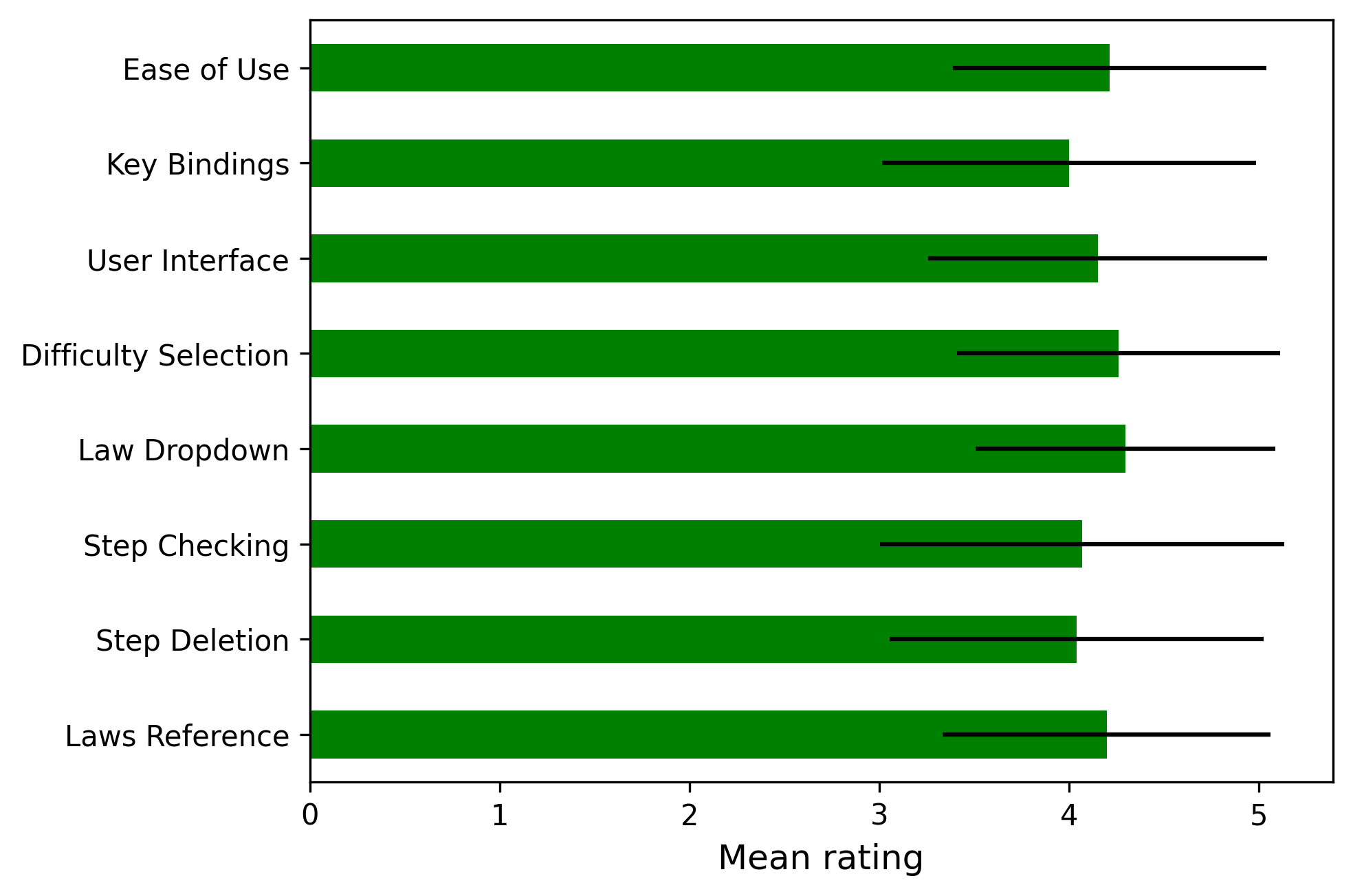}
  \caption{Aggregated mean feature ratings}
  \label{fig:afeat}
\end{minipage}%
\begin{minipage}{.5\textwidth}
  \centering
  \includegraphics[width=0.8\linewidth]{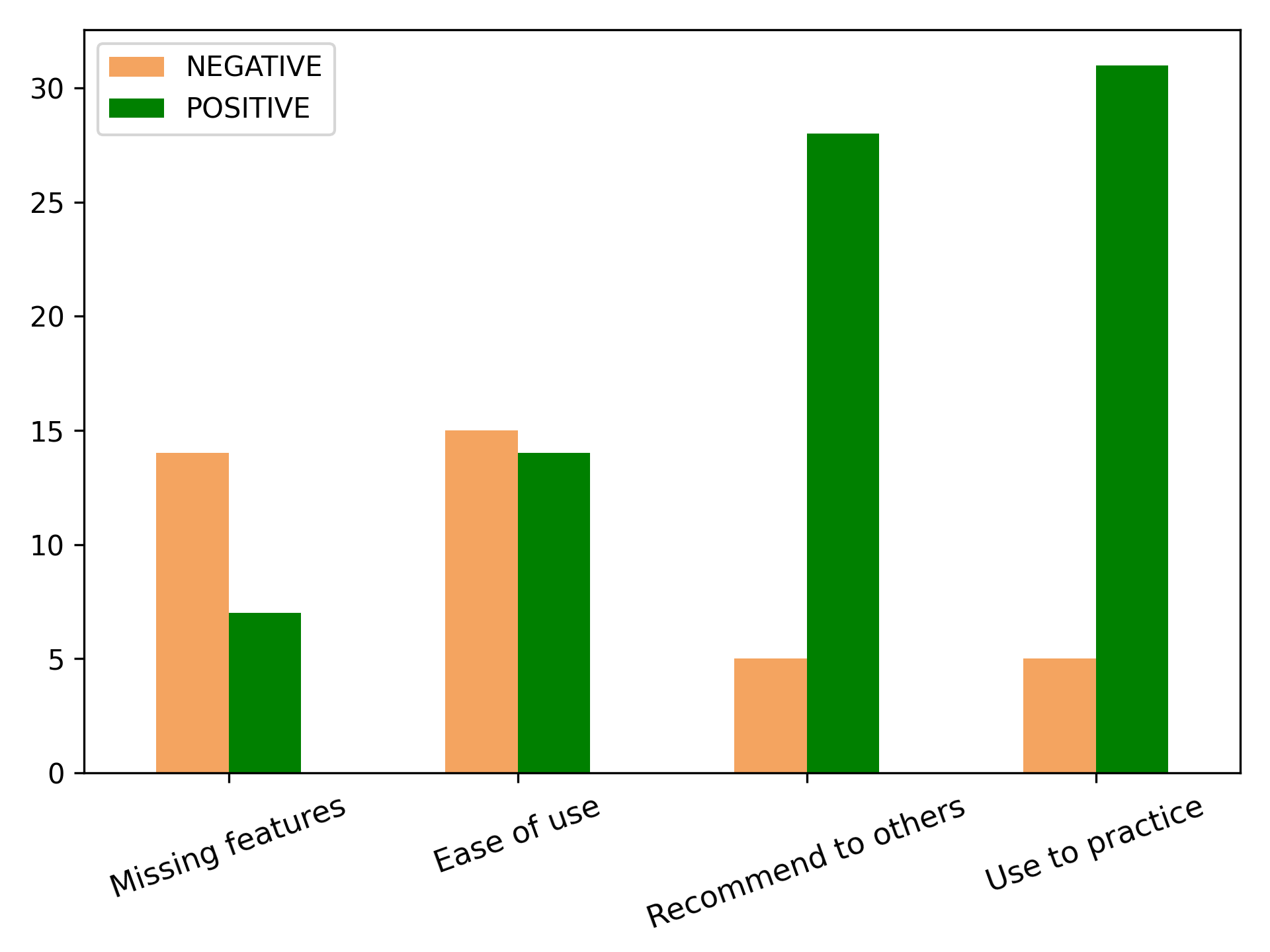}
  \caption{Sentiment to long-form questions}
  \label{fig:asentiment}
\end{minipage}
\end{figure}

For brevity, we only present the aggregate results of the surveys in these figures. Individual surveys (Appendix \ref{sec:apxSurvey}) are highly consistent, emphasizing the significance of our results. 

\subsection{Performance and AI results}\label{subsec:r2}

We use a genetic algorithm to optimize a weighted linear ensemble of heuristics (Table \ref{tab:heur}) for {\logiclearner} to solve logic proofs with A* search. The details of the optimization and ablation studies involved in our heuristic selection can be found in Appendix \ref{sec:apxAblate}. We compare {\logiclearner} with ChatGPT, a publicly available conversational interface to the GPT series of Large Language Models (LLMs) trained for text generation on an internet-scale dataset. ChatGPT is the most easily accessible AI question-answering tool to end users, and is increasingly being used in an academic setting. However, ChatGPT is not designed as a pedagogical tool, and the user experience varies significantly from that of {\logiclearner}. We find that it is also unable to produce mathematically consistent answers, and this \emph{confident but incorrect} behavior could lead to maladaptation from students learning incorrect information. Figure \ref{fig:flowCompare} describes one potential workflow of a student using ChatGPT for proof solving. In contrast with the {\logiclearner} user flow (Figure \ref{fig:uflow}), even when ChatGPT provides correct answers, it will not gradually nudge the user to a solution or take the user's progress into account unless prompted. Re-attempting and review are also not possible with such LLM-based QA systems.

\begin{figure}[ht!]
\centering
\includegraphics[width=0.8\textwidth]{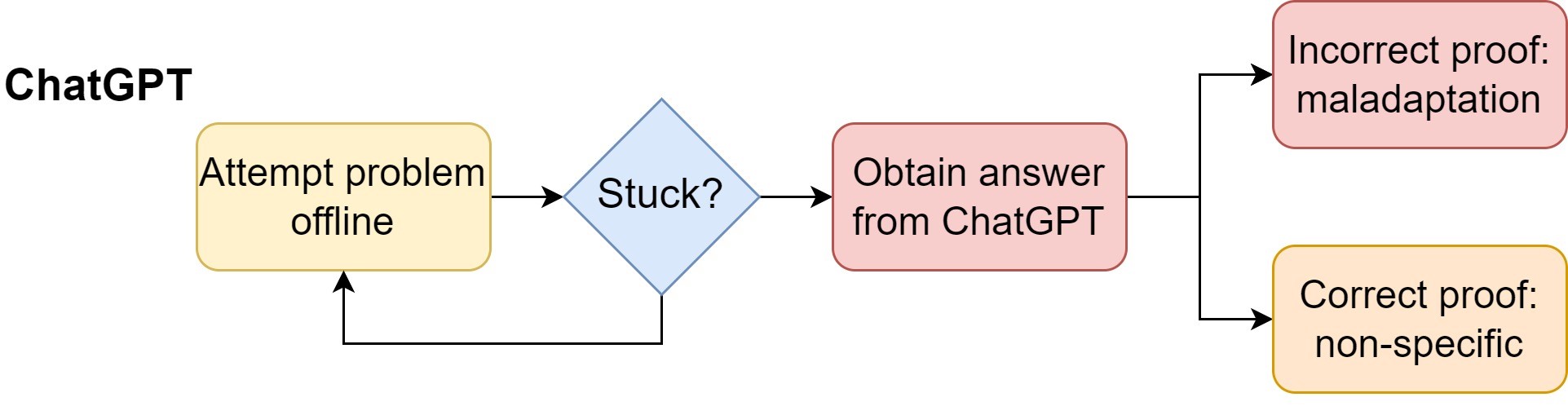}
\caption{A potential student workflow with ChatGPT for proof solving.}\label{fig:flowCompare}
\end{figure}

\subsubsection{Proof-solving performance}\label{subsubsec:r21}

We evaluate {\logiclearner} on the application's human-curated question bank of 33 logic proofs, as well as two QUESTION\_GEN (Algorithm \ref{alg:qgen}) generated datasets with 198 3-step proofs and 66 10-step proofs respectively. We also train and evaluate a simple neural network heuristic as a \emph{proof-of-concept} as described in Section \ref{subsec:m3}, and detail our experiments in Appendix \ref{sec:apxNN}. Due to resource constraints, we evaluate neural models (including ChatGPT) only on our human-curated dataset. We provide the following prompt to ChatGPT: \emph{`Using the rules of propositional logic, prove that \textbf{premise} is logically equivalent to \textbf{target}. Make sure to explain each step.'}. This prompt is manually tuned, and encourages some Chain-of-Thought reasoning \cite{wei2022chain}, though we do not claim to fully use this method. The LLM's responses are well-framed, but rarely correct. A student is likely not sophisticated enough to catch these technical mistakes, making ChatGPT infeasible as a learning aid for proof solving.

\begin{table}[ht!]
\begin{center}
\begin{minipage}{\textwidth}
\caption{AI performance on logic proof question banks.}\label{tab:aiRes}
\begin{tabular}{@{}llcccc@{}}
\toprule
Method & Heuristic & Eval time & \multicolumn{3}{c}{Score on Dataset (questions)} \\
 & & (seconds) & Curated & \multicolumn{2}{c}{Generated} \\ 
 & & & (33 mixed) & (198 small) & (66 mixed) \\
\midrule
\multirow{2}{*}{{\logiclearner}} & Comparator ensemble & 3 & \textbf{28} & \textbf{143} & \textbf{43} \\
 & Neural Network & 15 & 13 & - & - \\
\midrule
ChatGPT & Relevant prompt & 15 - 30 & 5 & - & -\\
\botrule
\end{tabular}
\end{minipage}
\end{center}
\end{table}

As seen in Table \ref{tab:aiRes}, {\logiclearner} is able to solve almost every question in our question bank in real time, providing a seamless user experience to students requesting hints. The {\logiclearner} heuristic is also able to generalize well to previously unseen and computer-generated questions of varying lengths, reinforcing its strength as a heuristic. In contrast, ChatGPT is unable to solve any proofs that are longer than 1-2 steps. Our inspection of ChatGPT's proofs shows some recurring error modes:
\begin{itemize}
    \item ChatGPT does not model the semantics of Boolean variables, resulting in false equivalences. For example, it does not distinguish between $q$ and $r$ in the bottom of Figure \ref{fig:chatFails}(a).
    \item It often hallucinates variable or rule names, even for correct equivalences. Figure \ref{fig:chatFails}(b) shows instances of ChatGPT confusing the Absorption and Domination laws, and incorrectly spelling `Commutative'.
    \item As seen in Figure \ref{fig:chatFails}(c), ChatGPT fails to correctly perform parity-dependent operations like negation and matching parentheses. Parity matching is known to be a hard problem for Transformer \cite{vaswani2017attention} based architectures like LLMs \cite{Hahn_2020}, and is one of the most frequent failure modes we observed.
    \item ChatGPT cannot accurately parse long expression sequences or maintain context over the span of a long answer, resulting in it hallucinating a conclusion to an incorrect target (Figure \ref{fig:chatFails}(d)).
\end{itemize} 

Additional details of our ChatGPT analysis (such as prompt selection) are presented in Appendix \ref{sec:apxChat}. We hypothesize that using {\logiclearner}'s proofs for Retrieval-Augmented Generation \citep{lewis2020retrieval} to LLMs could produce well-explained, mathematically sound proofs, but we leave this to future work.

\begin{figure*}[ht!]
    \centering
    \begin{subfigure}[t]{0.5\textwidth}
        \centering
        \includegraphics[width=\linewidth]{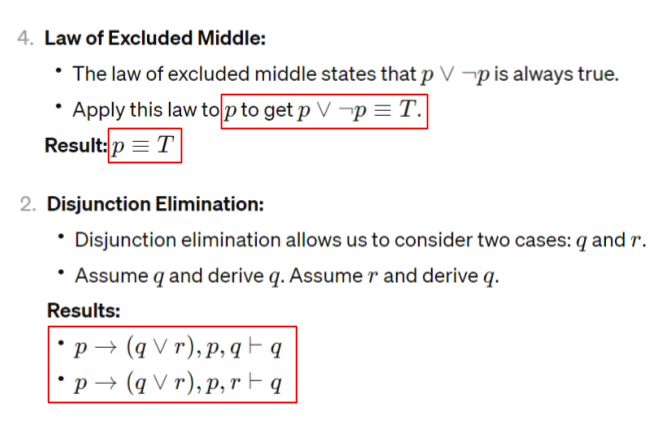}
        \caption{}
    \end{subfigure}%
    ~ 
    \begin{subfigure}[t]{0.5\textwidth}
        \centering
        \includegraphics[width=\linewidth]{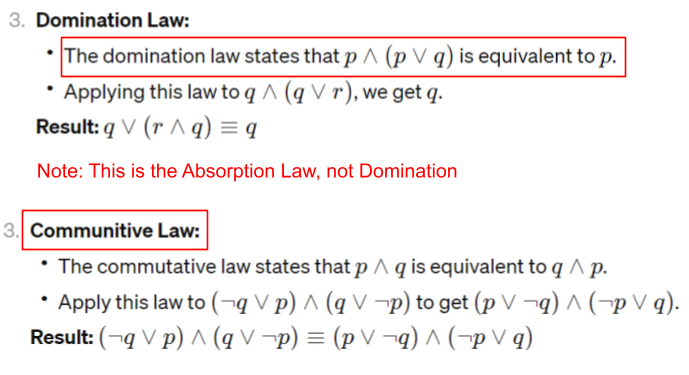}
        \caption{}
    \end{subfigure}%

    \begin{subfigure}[t]{0.5\textwidth}
        \centering
        \includegraphics[width=\linewidth]{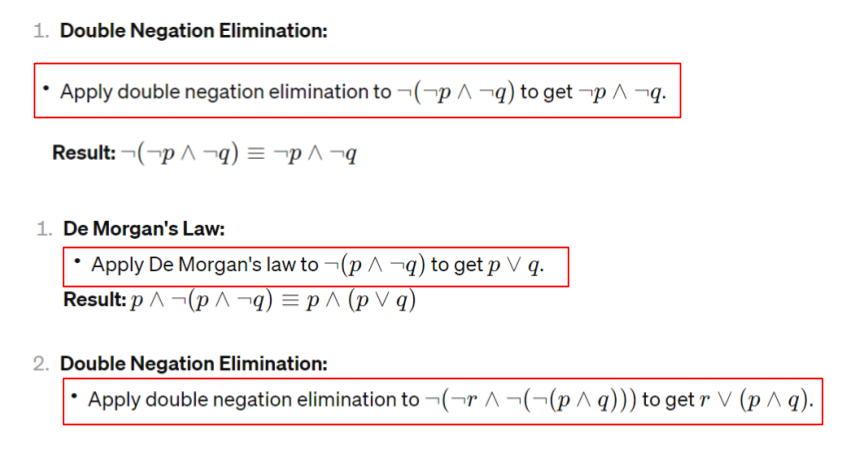}
        \caption{}
    \end{subfigure}%
    ~ 
    \begin{subfigure}[t]{0.5\textwidth}
        \centering
        \includegraphics[width=\linewidth]{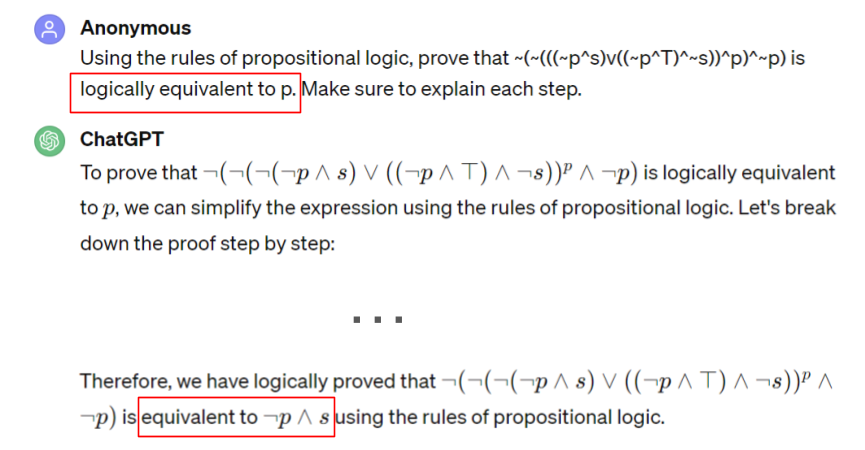}
        \caption{}
    \end{subfigure}%
    \caption{A variety of ChatGPT's failure modes on logic proofs.}\label{fig:chatFails}
\end{figure*}

\section{Limitations and future work}\label{sec:limit}

Our surveys across two semesters of undergraduate discrete mathematics show a strongly positive response to {\logiclearner}. However, voluntary student surveys are susceptible to only receiving responses from conscientious students who tend to perform well regardless of training method. Additionally, it is challenging to disentangle the benefits of {\logiclearner} from unaided student learning, since denying access to a practice tool a control group of students is unfair to that group. More concrete results would likely require a multi-year, multi-university effort. While {\logiclearner} is guaranteed to be mathematically precise, we use a simple AI heuristic that does not guarantee accuracy. Future work could explore the potential of large models as heuristics to {\logiclearner}'s graph search framework for proof solving---we explore such ideas in Appendix \ref{sec:apxNN}. {\logiclearner} is open-source, and we encourage the development of more sophisticated search methods. We also note that such improvements may make for interesting student projects.

\section{Conclusion}\label{sec:conclusion}

The study of propositional logic benefits from judgment-free guidance during practice. Building on existing work on the inability of Large Language Models (LLMs) to perform deep reasoning, we analyze the failure modes of ChatGPT, an LLM-based application that is increasingly used for pedagogical guidance. We outline the requirements of an effective practice tool for proof solving, and accordingly develop {\logiclearner}, an interactive application to practice propositional logic proofs with real-time an AI proof solver for real-time pedagogical guidance. We pilot {\logiclearner} as a practice tool over two semesters of the undergraduate discrete mathematics course at Columbia University and receive strongly positive feedback on its design and utility in student surveys. To the best of our knowledge, {\logiclearner} is the first and only application for fully-automated guided proof-solving practice at the time of writing. {\logiclearner} is free and open-source, and we look forward to the continued development of tools that democratize access to mathematical understanding.

\backmatter

\bmhead{Acknowledgements}

We thank the Center for Teaching and Learning (CTL) at Columbia University for lending their expertise in designing {\logiclearner} as an effective pedagogical tool, and for continuing to manage the {\logiclearner} application server and its open-source code base.

\section*{Declarations}

\begin{itemize}
\item \textbf{Funding} This work was supported by a generous grant from the Columbia University Provost's Faculty Committee on Educational Innovation. 
\item \textbf{Competing interests} All authors declare that they have no conflicts of interest in relation to this work.
\item \textbf{Ethics approval and consent to participate} Student surveys were approved by the Institutional Review Board. Approval number: IRB-AAAU0354.
\item \textbf{Consent for publication} All authors consent to the publication of this manuscript.
\item \textbf{Data availability} Student survey data is presented in aggregate form in Appendix \ref{sec:apxSurvey}. Full experimental results are in Appendix \ref{sec:apxAblate}.
\item \textbf{Materials availability} Not applicable.
\item \textbf{Code availability} {\logiclearner} is a free and open-source application under the GNU GPLv3 license. Source code is available at \url{https://github.com/ccnmtl/logiclearnertools}.
\item \textbf{Author contribution} AI implemented the business logic (Parser, Search, API), designed and conducted the technical experiments, and authored the manuscript. MV and UM designed and implemented an earlier version of the {\logiclearner} parser. The staff at the CTL (MT, ZM, ND, SS) designed the {\logiclearner} web application user experience and database. The application server and git repositories are managed by the CTL. As the PIs of this project, NV and AS developed the vision for {\logiclearner}, oversaw the entirety of its design and implementation, and mentored all student contributors who worked on the application.
\end{itemize}

\clearpage

\bibliography{sn-bibliography}

\clearpage

\begin{appendices}

\section{Visuals of the LogicLearner Web Application}\label{sec:apxVisual}

We designed the user interface and user experience of {\logiclearner} in collaboration with the Columbia University Center for Teaching and Learning\footnote{\url{https://ctl.columbia.edu}} (CTL). The CTL specializes in designing the user experience of educational tools, and hosts the {\logiclearner} application and database on its server. 

\begin{figure}[ht!]
\includegraphics[width=0.99\textwidth, frame]{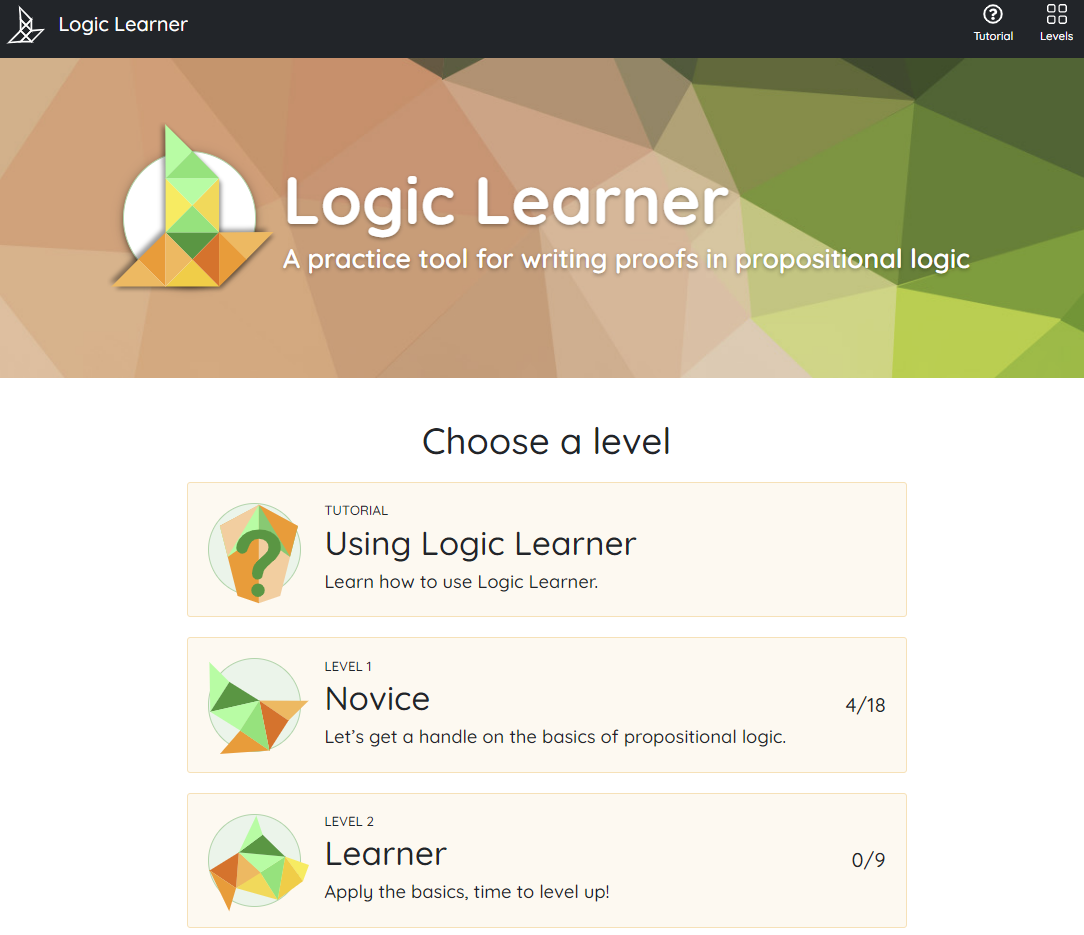}
\caption{The {\logiclearner} home screen.}\label{fig:appIntro}
\end{figure}

\begin{figure}[ht!]
\includegraphics[width=0.99\textwidth, frame]{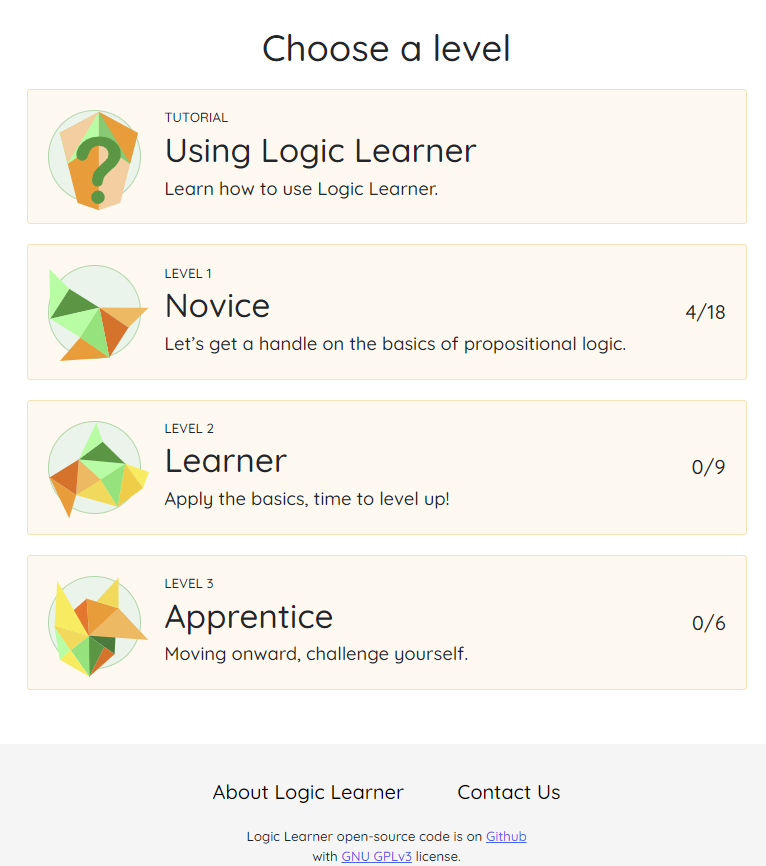}
\caption{Level selection according to proof difficulty.}\label{fig:appLevel}
\end{figure}

\begin{figure}[ht!]
\includegraphics[width=0.99\textwidth, frame]{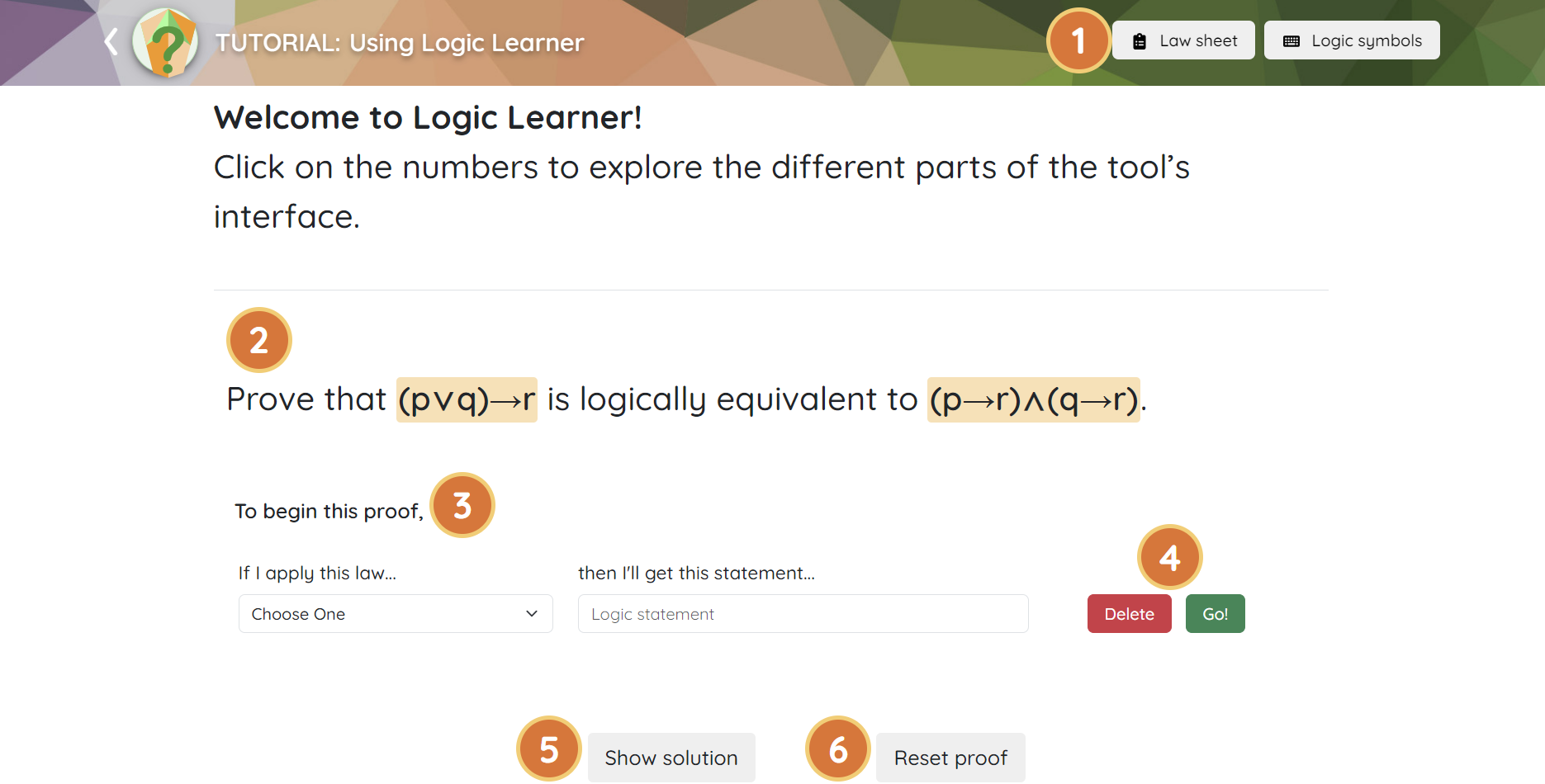}
\caption{A tutorial on using {\logiclearner}.}\label{fig:appGuide}
\end{figure}

\begin{figure}[ht!]
\includegraphics[width=0.99\textwidth, frame]{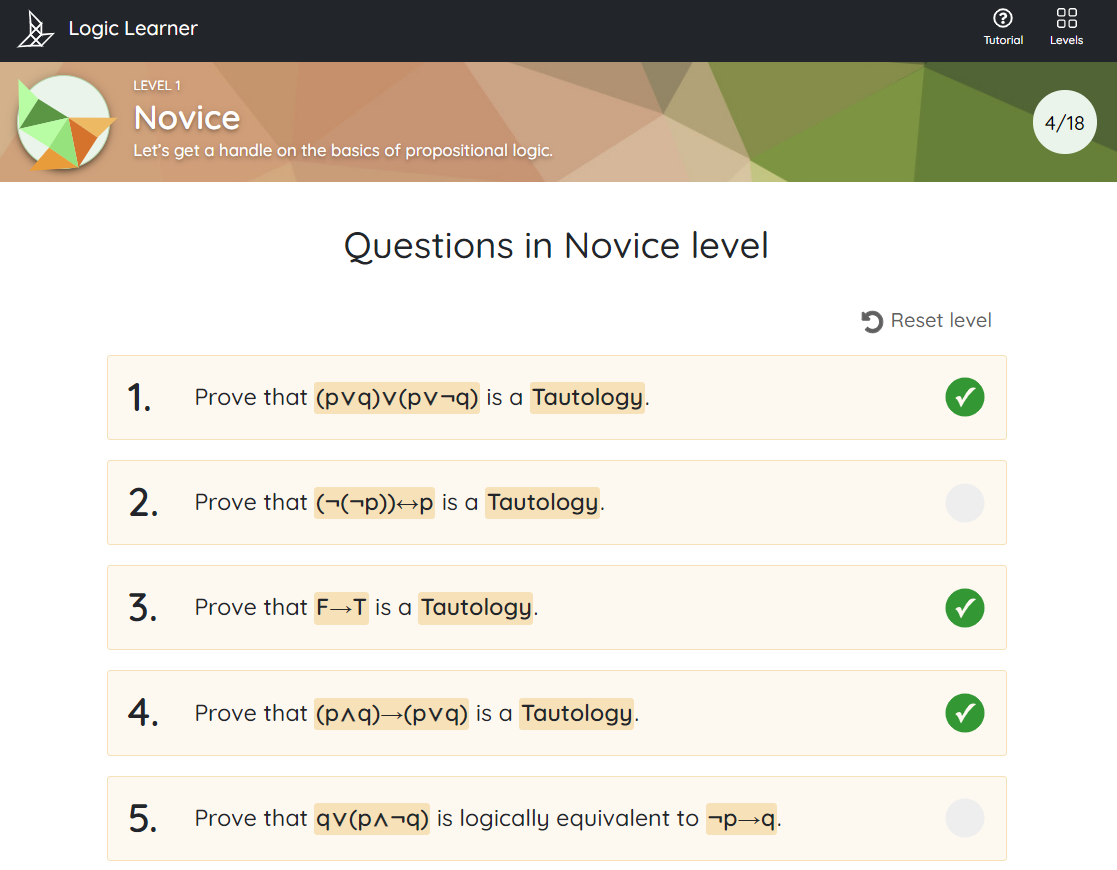}
\caption{Selecting questions in the Novice level.}\label{fig:appNovice}
\end{figure}

\begin{figure*}[t!]
    \centering
    \begin{subfigure}[t]{0.48\textwidth}
        \centering
        \includegraphics[width=\linewidth]{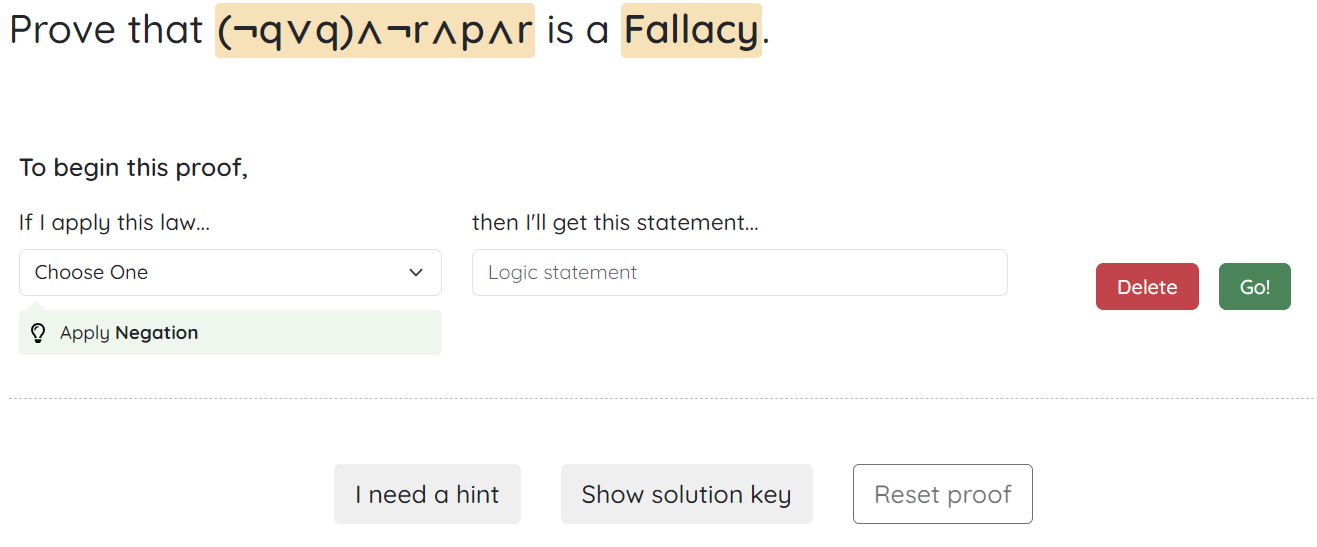}
        \caption{One hint requested.}
    \end{subfigure}%
    ~ 
    \begin{subfigure}[t]{0.5\textwidth}
        \centering
        \includegraphics[width=\linewidth]{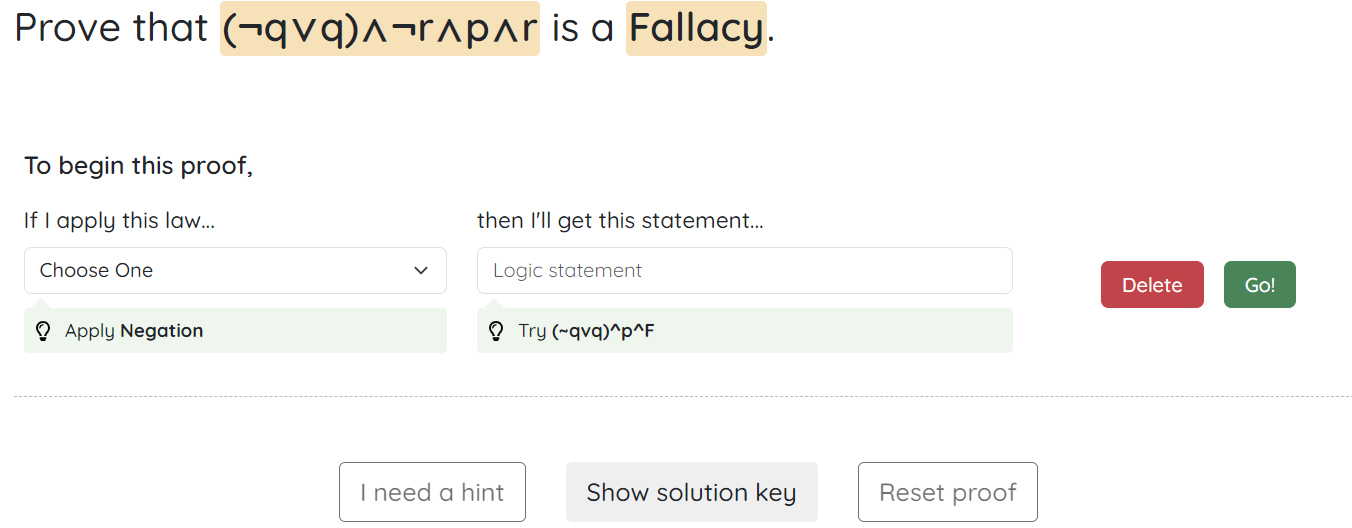}
        \caption{Two hints requested}
    \end{subfigure}
    \caption{Requesting hints during a question attempt.}
\end{figure*}

\begin{figure}[ht!]
\includegraphics[width=0.99\textwidth]{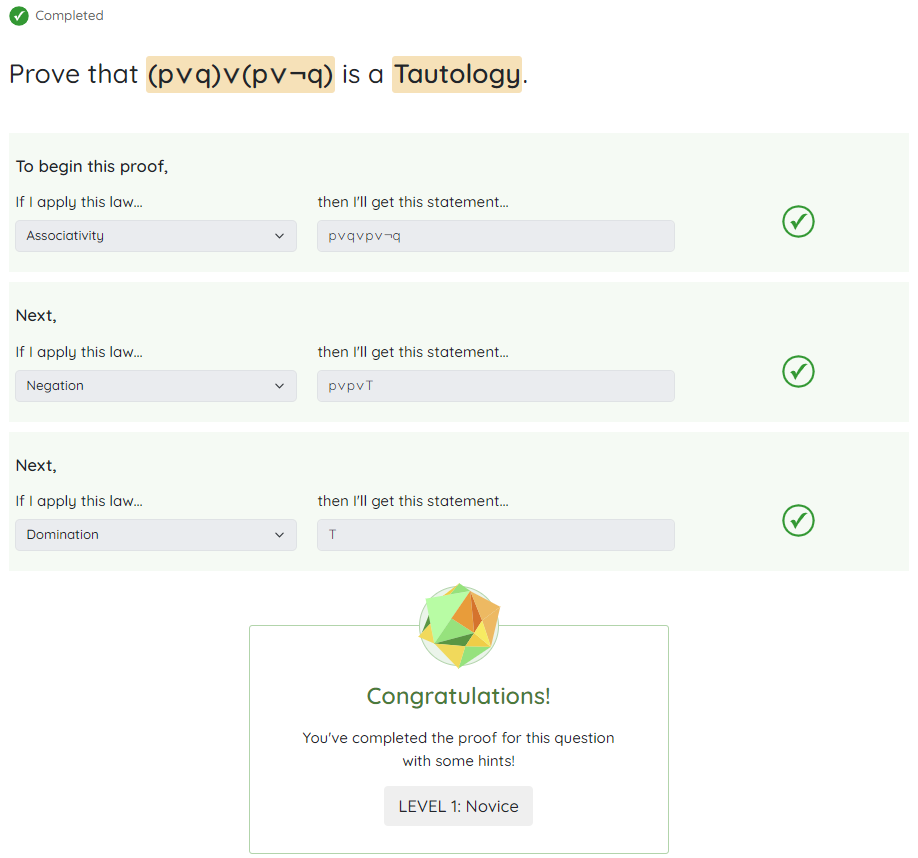}
\caption{A completed logic proof.}\label{fig:appCompleted}
\end{figure}

\section{AI Experiments and Ablations}\label{sec:apxAblate}

We train a genetic algorithm to optimize the weights of a linear combination of our heuristics defined in Table \ref{tab:heur}. Weights were constrained to floating point numbers in $[-10,10]$, which produced the best results empirically. The hyperparameters explored were the per-question time $\tau$, number of training questions $|Q|$, elitism $\epsilon$, and the crossover and mutation probabilities $P_c,P_m$. A score (questions solved) was calculated for three datasets: human-curated $(S_v)$, short generated questions ($S_t$), and mixed-length generated questions $(S_l)$. We choose a small subset of our human-curated question bank to train our algorithm, and present the final performance on all questions in \ref{tab:ablate1}. The weights in our production heuristic are presented in Table \ref{tab:prodheur}. Note that this heuristic was committed to production prior to some of the later AI improvements. 

Experiments were conducted on modest hardware with a limited search depth, and better hardware may lead to better performance. Despite this, our genetic algorithm shows a considerable improvement over the random baseline.

\begin{table}[ht!]
\begin{center}
\begin{minipage}{\textwidth}
\caption{Genetic algorithm ablation trials}\label{tab:ablate1}
\begin{tabular}{@{}llllllllll@{}}
\toprule
Population & Total & $\tau$ & $|Q|$ & $\epsilon$ & $P_c$ & $P_m$ & $S_v$ & $S_t$ & $S_l$ \\
Size & Generations & (s) & & & & & (33) & (198) & (66) \\ 
\midrule
5 & 5 & 3 & 5 & 1 & 0.8 & 0.2 & 22 & 119 & 29 \\
10 & 10 & 3 & 10 & 1 & 0.8 & 0.2 & 21 & 114 & 32 \\
10 & 10 & 3 & 20 & 1 & 0.8 & 0.2 & 27 & 129 & 38 \\
10 & 10 & 3 & 33 & 1 & 0.8 & 0.2 & 27 & 141 & 43 \\
10 & 10 & 5 & 33 & 1 & 0.8 & 0.2 & 26 & 135 & 40 \\
10 & 10 & 3 & 33 & 3 & 0.8 & 0.2 & \textbf{29} & 137 & 42 \\
10 & 20 & 3 & 33 & 3 & 0.6 & 0.2 & 28 & \textbf{143} & \textbf{43} \\
20 & 10 & 3 & 33 & 1 & 0.8 & 0.2 & 28 & 140 & 43 \\
20 & 20 & 3 & 33 & 1 & 0.8 & 0.2 & 28 & 133 & 43 \\
20 & 10 & 33 & 30 & 3 & 0.8 & 0.2 & 27 & 127 & 40 \\
\midrule
Production & & & & & & & & &  \\
\midrule
20 & 8 & 1 & 7 & 3 & 0.8 & 0.5 & 27 & 129 & 37 \\
\botrule
\end{tabular}
\end{minipage}
\end{center}
\end{table}

\begin{table}[ht!]
\begin{center}
\begin{minipage}{\textwidth}
\caption{A* search heuristics}\label{tab:prodheur}%
\begin{tabular}{@{}ll|ll@{}}
\toprule
Heuristic & Weight & Heuristic & Weight \\
\midrule
Levenshtein distance & 3.36 & Distributivity & 3.94 \\
Unitary function & 3.76 & Domination & 4.09 \\
Variable Mismatch & 6.09 & Idempotence & -7.03 \\
Length difference & 1.53 & Identity & -9.85 \\
Absorption & -3.88 & Iff as Implication & -4.20 \\
Associativity & 1.94 & Implication as Disjunction & 6.92 \\
Commutativity & -8.07 & Negation & -0.55 \\
De Morgan's Law & 3.71 & Start state & 1.44 \\
\botrule
\end{tabular}
\end{minipage}
\end{center}
\end{table}

We also trained our model on our dataset of 66 randomly generated solutions of depths 2-10 evaluated performance on the original question bank. These results, presented in Table \ref{tabR2}, show that our heuristic generalizes from our randomly generated dataset to the human-generated question bank.

\begin{table}[ht!]
\begin{center}
\begin{minipage}{\textwidth}
\caption{Train on randomly generated questions}\label{tabR2}
\begin{tabular}{@{}lllllllllll@{}}
\toprule
Population & Total & $\tau$ & $|Q|$ & $\epsilon$ & $\tau_\mathrm{eval}$ & $P_c$ & $P_m$ & $S_v$ & $S_t$ & $S_l$ \\
Size & Generations & (s) & & & (s) & & & (33) & (198) & (66) \\ 
\midrule
5 & 10 & 3 & 5 & 1 & 1 & 0.8 & 0.3 & 25 & 136 & 41 \\
5 & 10 & 3 & 5 & 1 & 3 & 0.8 & 0.3 & 27 & 143 & 42 \\
10 & 10 & 3 & 10 & 1 & 1 & 0.8 & 0.3 & 28 & 140 & 43 \\
10 & 10 & 3 & 10 & 1 & 3 & 0.8 & 0.3 & 29 & 147 & 46 \\
10 & 10 & 3 & 10 & 1 & 1 & 0.8 & 0.5 & 26 & 141 & 45 \\
10 & 10 & 3 & 10 & 1 & 3 & 0.8 & 0.5 & 27 & 143 & 45 \\
10 & 10 & 3 & 10 & 1 & 1 & 0.9 & 0.3 & 27 & 126 & 41 \\
10 & 10 & 3 & 10 & 1 & 3 & 0.9 & 0.3 & 29 & 138 & 42 \\
20 & 5 & 3 & 66 & 1 & 1 & 0.8 & 0.3 & 27 & 131 & 44 \\
20 & 5 & 3 & 66 & 1 & 3 & 0.8 & 0.3 & 29 & 139 & 45 \\
20 & 10 & 3 & 66 & 1 & 1 & 0.8 & 0.3 & 28 & 142 & 41 \\
20 & 10 & 3 & 66 & 1 & 3 & 0.8 & 0.3 & 29 & 147 & 44 \\
\botrule
\end{tabular}
\end{minipage}
\end{center}
\end{table}

\section{Student Survey Results}\label{sec:apxSurvey}

\begin{figure}[ht!]
\includegraphics[width=\textwidth]{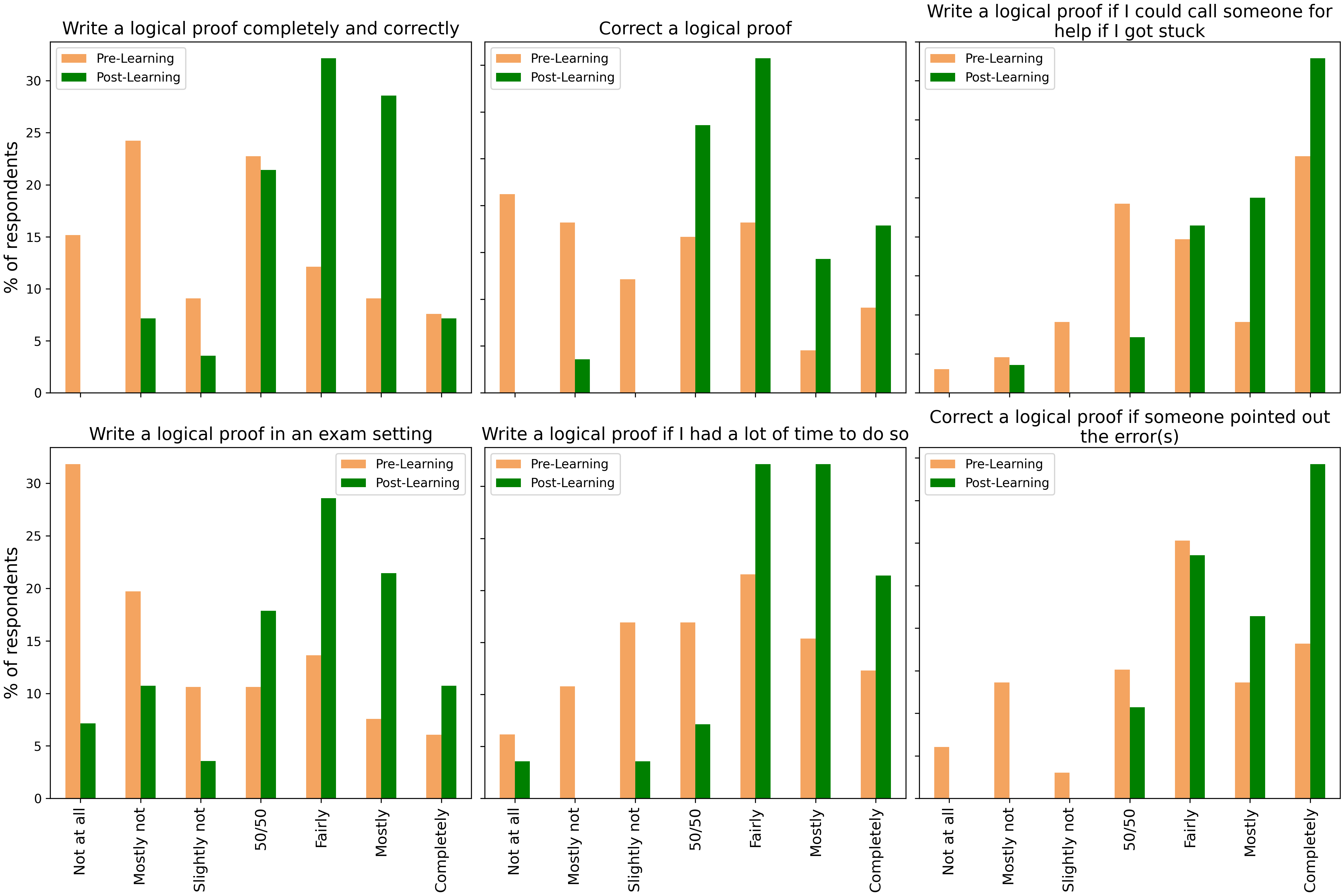}
\caption{Survey 1: confidence scores}\label{fig:s1conf}
\end{figure}

\begin{figure}[ht!]
\includegraphics[width=\textwidth]{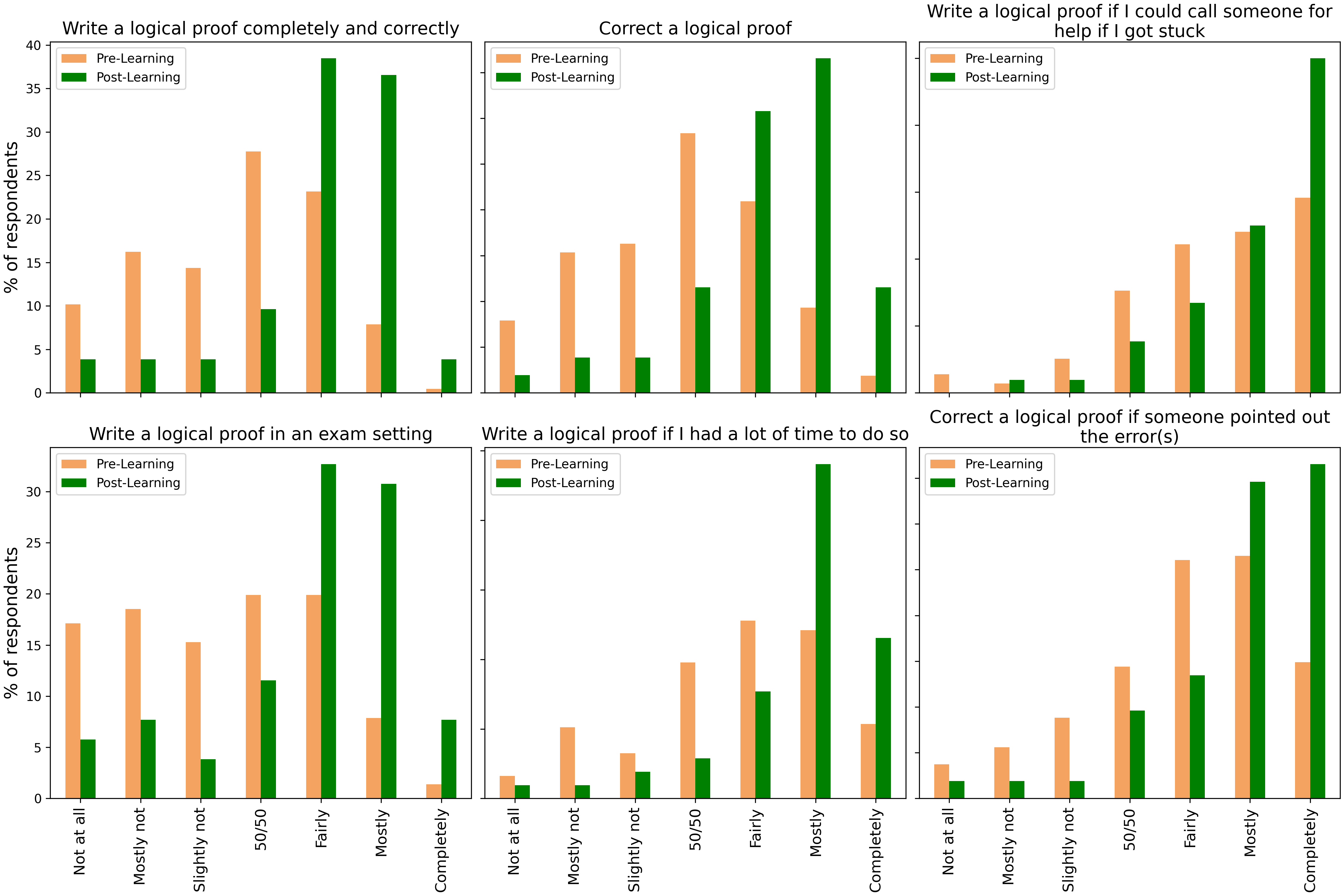}
\caption{Survey 2: confidence scores}\label{fig:s2conf}
\end{figure}

As seen in Figures \ref{fig:s1conf} and \ref{fig:s2conf}, the confidence scores are highly similar across semesters. There were no changes to the {\logiclearner} application in production between evaluations.

\begin{figure}[ht!]
\centering
\begin{minipage}{.5\textwidth}
  \centering
  \includegraphics[width=0.9\linewidth]{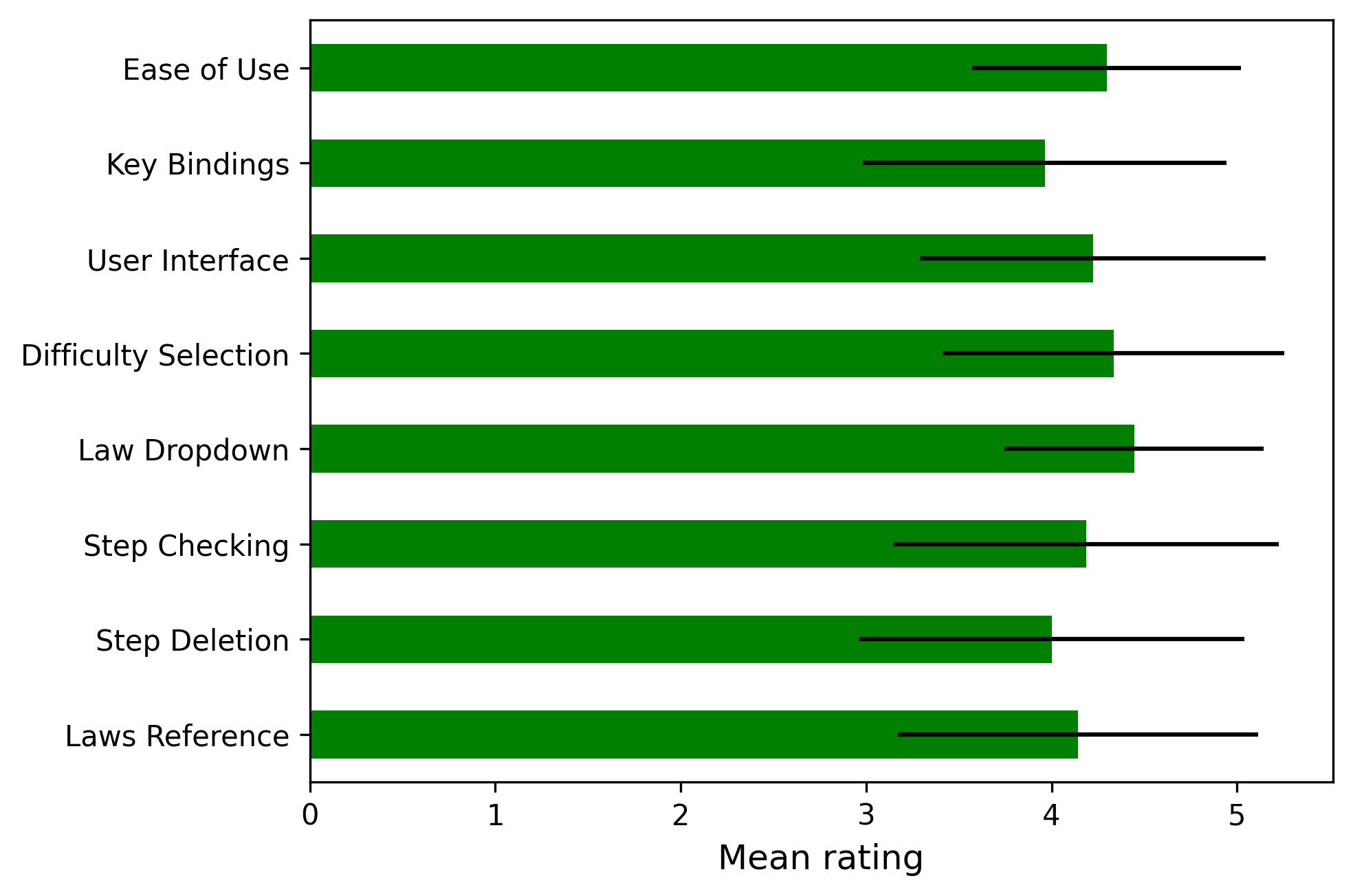}
  \caption{Survey 1: feature ratings}
  \label{fig:s1feat}
\end{minipage}%
\begin{minipage}{.5\textwidth}
  \centering
  \includegraphics[width=0.8\linewidth]{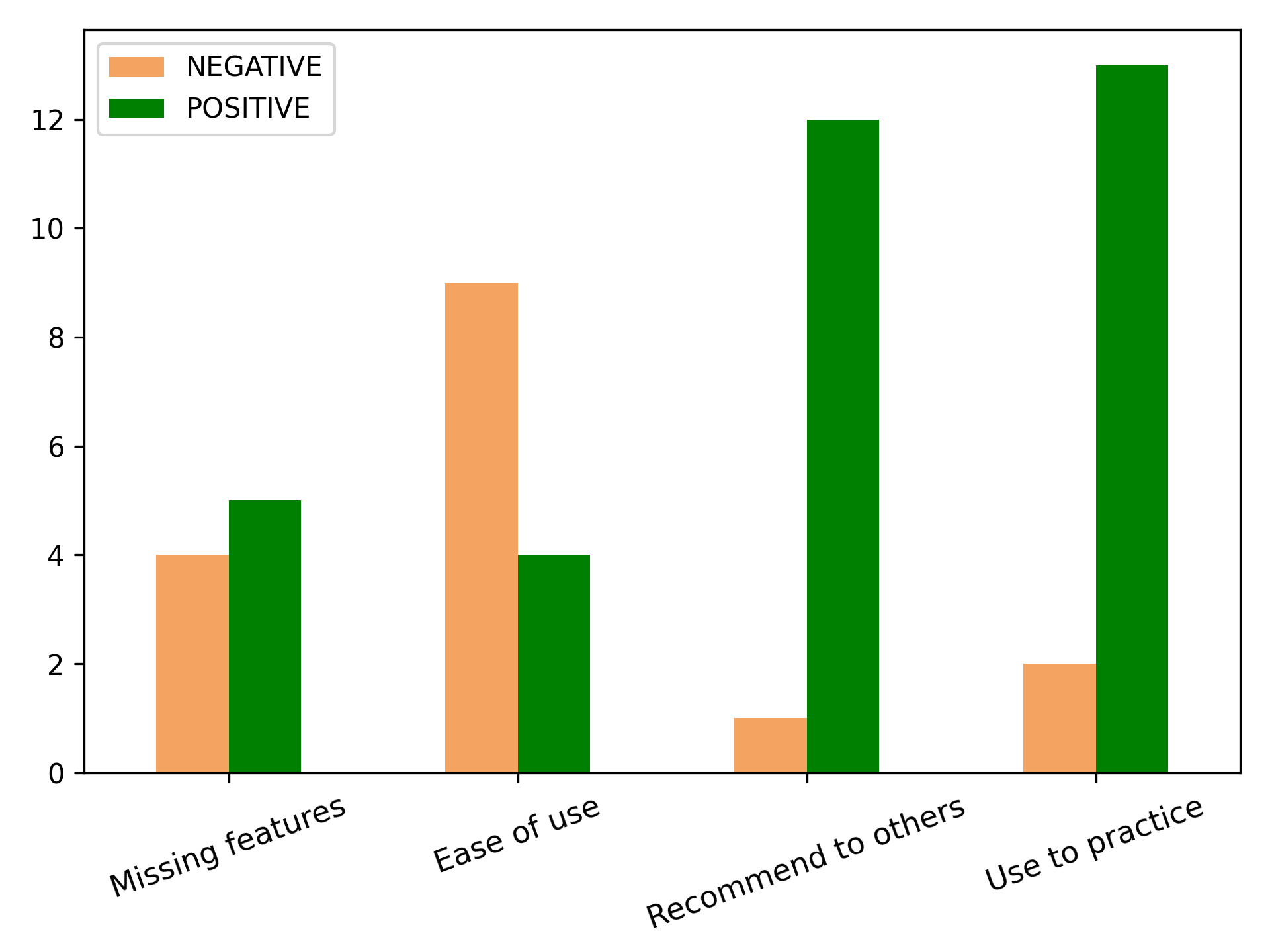}
  \caption{Survey 1: long-form sentiment}
  \label{fig:s1sentiment}
\end{minipage}
\end{figure}

\begin{figure}[ht!]
\centering
\begin{minipage}{.5\textwidth}
  \centering
  \includegraphics[width=0.9\linewidth]{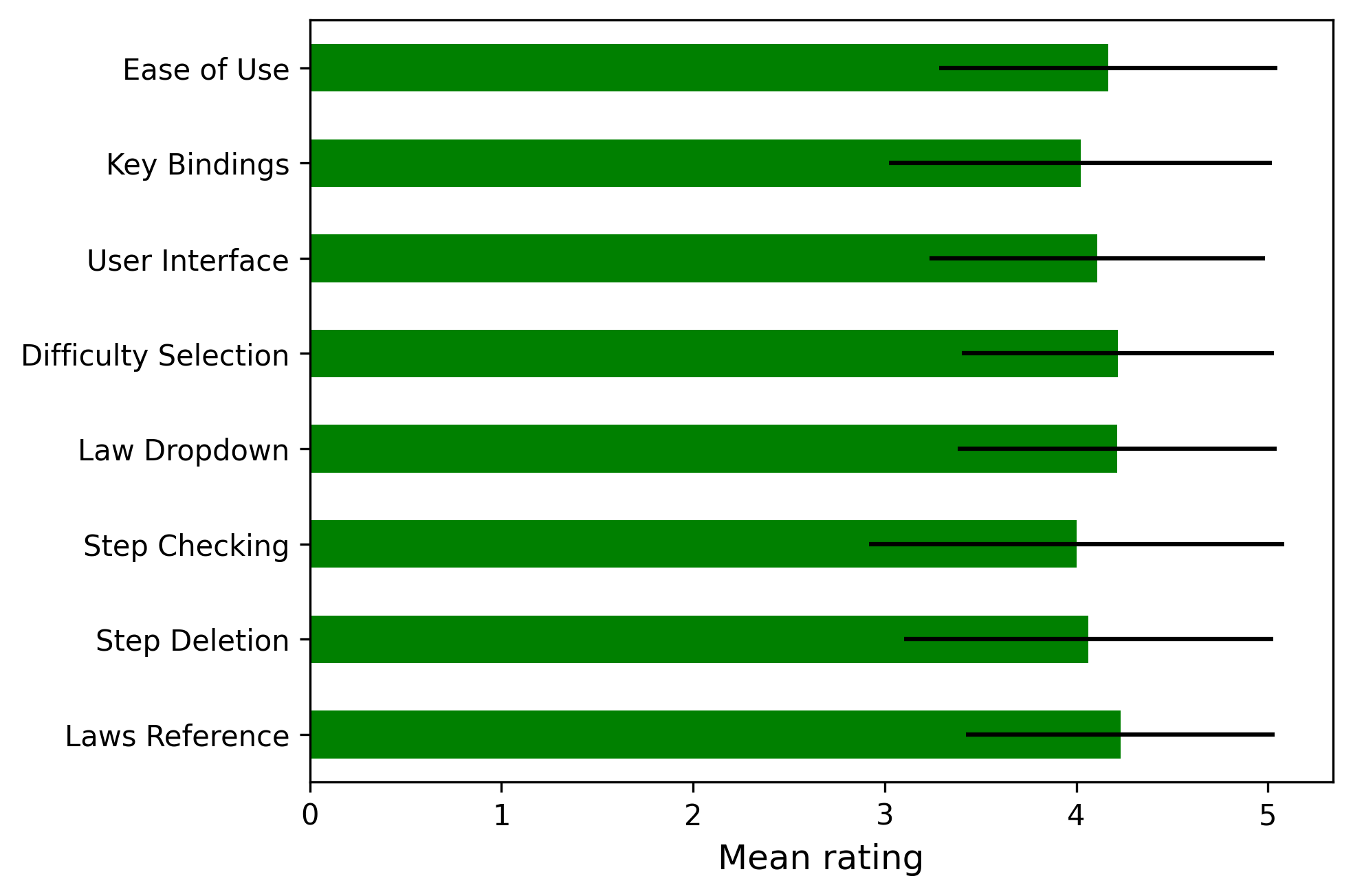}
  \caption{Survey 2: feature ratings}
  \label{fig:s2feat}
\end{minipage}%
\begin{minipage}{.5\textwidth}
  \centering
  \includegraphics[width=0.8\linewidth]{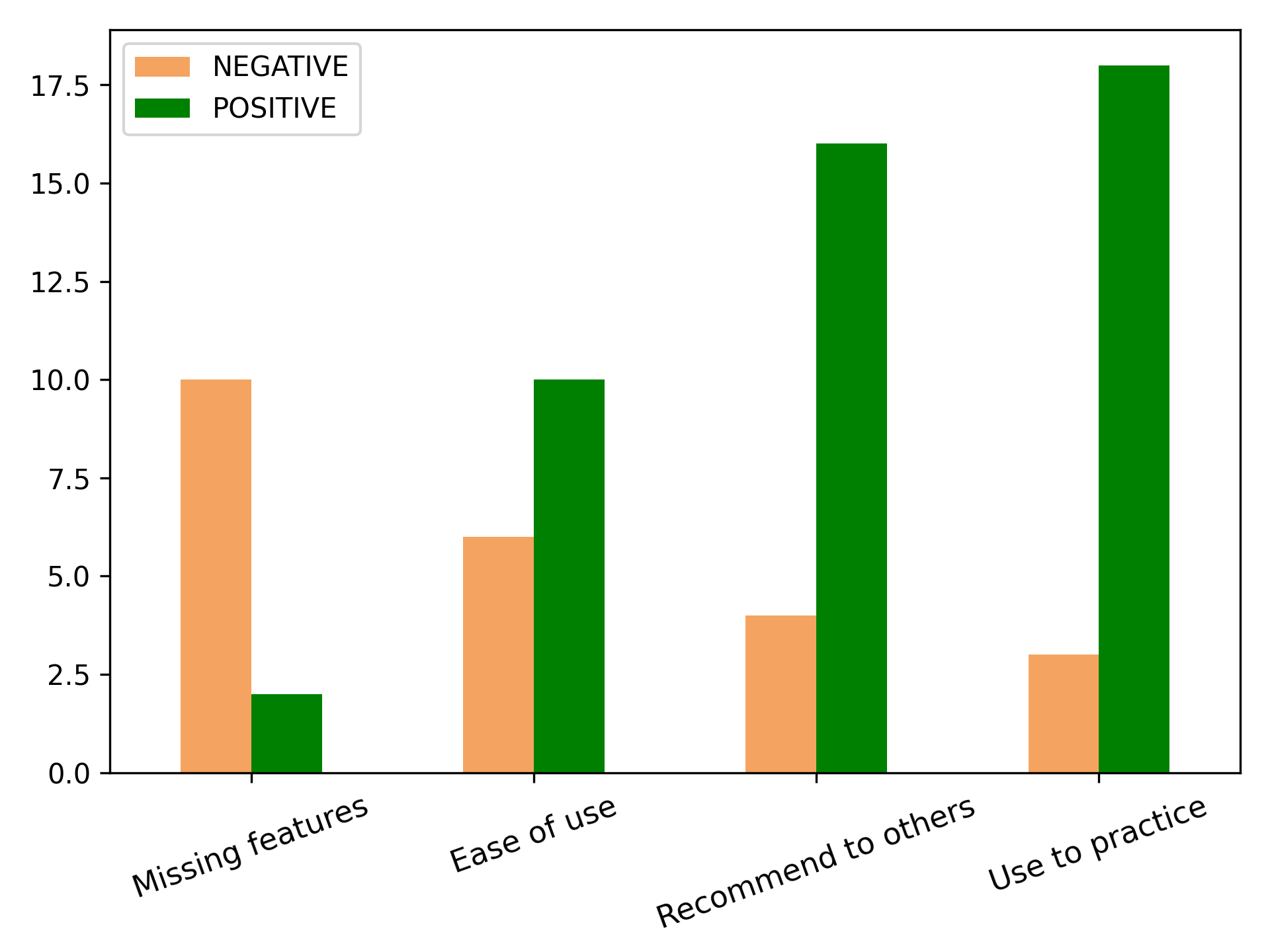}
  \caption{Survey 2: long-form sentiment}
  \label{fig:s2sentiment}
\end{minipage}
\end{figure}

The feature scores and long-form sentiment (Figures \ref{fig:s1feat}, \ref{fig:s2feat}, \ref{fig:s1sentiment}, \ref{fig:s2sentiment}) show a similar consistency across surveyed groups, bolstering our confidence in these results.

Table \ref{tab:longform} shows the number of respondents per question, with the long-form questions written as queried. The number of respondents is significantly higher in the pre-learning, but there are enough respondents (especially when aggregated) to show significant results for the post-learning evaluation. 

\begin{table}[ht!]
\begin{center}
\begin{minipage}{\textwidth}
\caption{Survey questions and number of respondents}\label{tab:longform}
\begin{tabular}{@{}llll@{}}
\toprule
Question & Survey 1 & Survey 2 & Total \\
\midrule
Pre-learning confidence (6 questions) & 65 & 213 & 278 \\
Post-learning confidence (6 questions) & 28 & 52 & 80 \\
\midrule
Feature ratings (scale of 1-5) & 27 & 43 & 70 \\
\midrule
What features are we missing? & 10 & 13 & 23 \\ 
How could we improve our app? & & & \\
\midrule
How easy is it to use our app? & & & \\ 
What were some challenges you faced & 14 & 16 & 30 \\
in using our app? & & & \\
\midrule
Would you recommend Logic Learner & & & \\ 
to others for practicing propositional logic? & 14 & 20 & 34 \\
Why or why not? & & & \\
\midrule
Would you use this app for practicing & 15 & 21 & 36 \\ 
for an exam or homework? & & & \\
\botrule
\end{tabular}
\end{minipage}
\end{center}
\end{table}

\section{Deep Boolean Metric Learning}\label{sec:apxNN}

{\logiclearner}'s hints feature requires a robust and efficient heuristic for proof solving with A* search. In this section, we explore a potentially powerful alternative heuristic to {\logiclearner} AI. Deep metric learning \cite{kaya2019deep} aims to embed data in high-dimensional metric space such that desirable relationships between data points are preserved w.r.t. a chosen metric. Here, we embed Boolean expressions in vector space with the aim of preserving semantic similarity as a metric between expression embeddings. We leverage \textbf{PROOF\_GEN} (Algorithm \ref{alg:qgen}) to produce datasets for neural network training, and train small Siamese encoder-decoder neural networks on two pre-training $k$-way classification tasks, rule prediction and proof length estimation (Table \ref{tab:nnproxy}) with a 70-30 train-test split. We restrict the proof length to be at most 3 steps to reduce the impact of sub-optimal generated proofs. Since there is an infinite vocabulary of Boolean expressions, we tokenize and learn fixed vocabulary embeddings at a character level. We train a Gated Recurrent Unit (GRU) \cite{cho2014learning} encoder and a fully-connected multilayer perceptron (MLP) decoder for 350 epochs on a T4 GPU in Google Colab\footnote{Google. (2024). Google Colaboratory: \url{https://colab.research.google.com/}}. The networks achieve high scores in pre-training despite their simplicity, as shown in Figures \ref{fig:ruleTrain}, \ref{fig:distTrain} and Table \ref{tab:nnproxy}.

\begin{figure}[ht!]
\centering
\begin{minipage}{.5\textwidth}
  \centering
  \includegraphics[width=\linewidth]{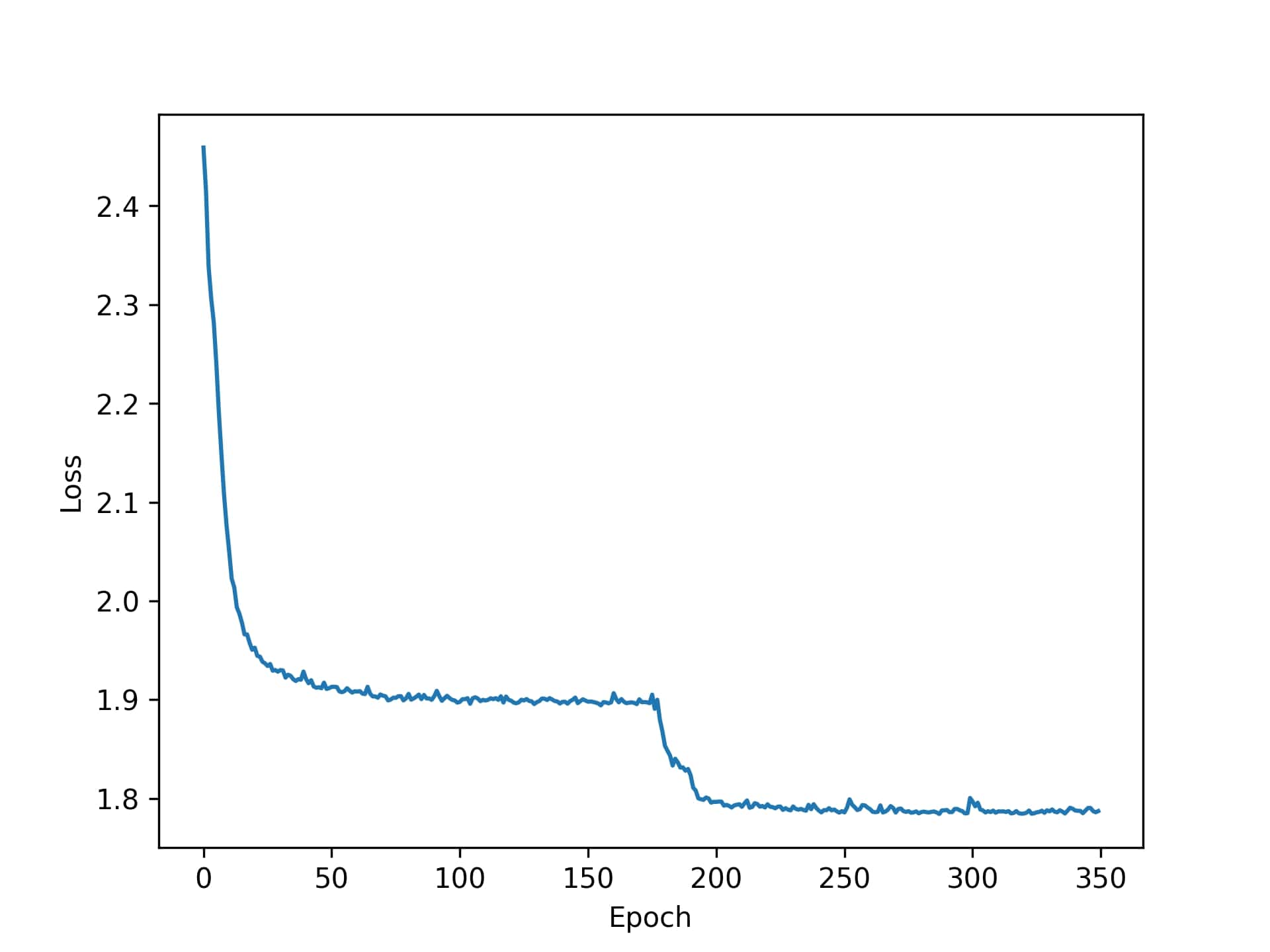}
  \caption{Training loss: rule prediction}
  \label{fig:ruleTrain}
\end{minipage}%
\begin{minipage}{.5\textwidth}
  \centering
  \includegraphics[width=\linewidth]{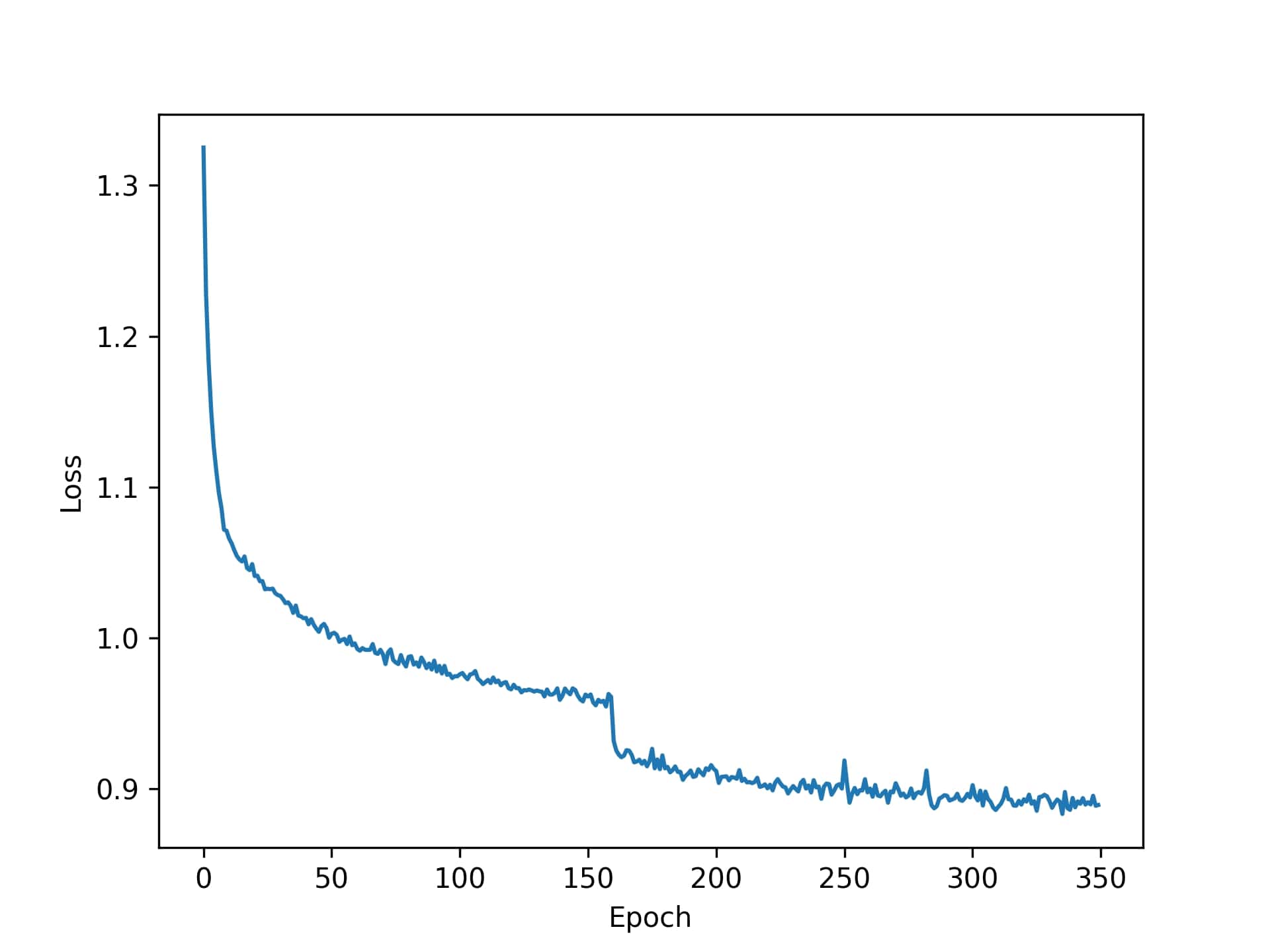}
  \caption{Training loss: proof length}
  \label{fig:distTrain}
\end{minipage}
\end{figure}

\begin{table}[ht!]
\begin{center}
\begin{minipage}{\textwidth}
\caption{Neural Network on Proxy Tasks}\label{tab:nnproxy}
\begin{tabular}{@{}lllccccc@{}}
\toprule
Pre-training & Input & Target & Dataset & Num. & \multicolumn{3}{c}{Top-1 Accuracy (\%)} \\
Task & & & Size & Class & Random & Train & Test \\
\midrule
Rule Prediction & $(e1, e2)$ & $r\in R: e2\in r(e1)$& 3128 & 16 & 6.25 & 83.10 & 81.04  \\
Proof Length & $(e1, e2)$ & Proof length $l\le 3$ & 16797 & 4 & 25 & 86.17 & 80.58 \\
\botrule
\end{tabular}
\end{minipage}
\end{center}
\end{table}

Post training on proxy tasks, we use the cosine similarity between GRU encoder outputs as the heuristic for A* search on our human-curated question bank. As a baseline, we use a randomly initialized GRU encoder without pre-training. As an ablation study on our choice of tokenizer and training data, we also use CANINE-s \cite{DBLP:journals/corr/abs-2103-06874} embeddings as a heuristic. CANINE-s is a Transformer-based language model that does not require an explicit tokenizer, bypassing the limitations of fixed-vocabulary models in parsing Boolean expressions.

\begin{table}[ht!]
\begin{center}
\begin{minipage}{\textwidth}
\caption{Neural Network on Proxy Tasks}\label{tab:nnresults}
\begin{tabular}{@{}llccc@{}}
\toprule
Model & Language & {\logiclearner} & Timeout & Score \\
 & Pre-training & Pre-training & (s) & (33) \\
\midrule
\multirow{6}{*}{GRU Encoder} & \multirow{6}{*}{None} & \multirow{2}{*}{Rule Prediction} & 5 & 7 \\
 & & & 15 & \textbf{13} \\
 & & \multirow{2}{*}{Proof Length} & 5 & 6 \\
 & & & 15 & 9 \\
 & & \multirow{2}{*}{None} & 5 & 7  \\
 & & & 15 & 10 \\
\midrule
\multirow{2}{*}{CANINE-s} & \multirow{2}{*}{Multilingual Wikipedia \cite{DBLP:journals/corr/abs-1810-04805}} & \multirow{2}{*}{None} & 5 & 2 \\
 & & & 15 & 2  \\
\botrule
\end{tabular}
\end{minipage}
\end{center}
\end{table}

Table \ref{tab:nnresults} shows the results of our neural network study. While pre-training on proof length prediction does not help in this case, rule prediction improves the performance at greater search depths. Our ablations show that much of the performance comes from the domain-specific tokenization and network architecture, as CANINE-s---a complex language model with extensive multilingual pre-training---performs significantly worse than our baseline. Additionally, all models except CANINE-s outperform ChatGPT. However, all models are outperformed by our ensemble in production. 
We emphasize again that this is merely a \textbf{proof-of-concept} for potential future works that could leverage {\logiclearner}'s proof-generation and AI structure to build powerful automated logic proof solvers.

\section{ChatGPT Struggles with Logic Proofs}\label{sec:apxChat}

To evaluate ChatGPT as a logic proof solver, we tune the prompts to produce answers that resemble step-by-step proofs using the rules of propositional logic\footnote{Full prompt tuning transcript: \url{https://chat.openai.com/share/9fc940a7-e258-4512-967e-c0d8e780ba8b}}. We do not use complex compositional prompts as our end users (undergraduate students) are unlikely to be familiar with the mathematical prompt tuning literature. Without explicitly asking for rules, ChatGPT generates a truth table to solve proofs. Our final prompt is of the form described in Section \ref{subsec:r2}. With this simple prompt, we ask ChatGPT all questions in our human-curated question bank\footnote{Full proof-evaluation transcript: \url{https://chat.openai.com/share/7b7aa7f5-fb66-4fd3-ae41-172e2bf257d2}}.

While ChatGPT exhibits an astonishing ability to parse input and perform few-step reasoning, it exhibits several error modes with our simple but realistic prompt, detailed in Section \ref{subsec:r2}. ChatGPT interprets the term `rules of logic' to beyond equivalence relations include operations such as Modus Ponens, but we do not count it against the results if the proof is logically correct. We ask it each question once as a naive user will not retry in case of an incorrect answer. By our (lenient) analysis of the results, ChatGPT only solves the 5 questions in Table \ref{tab:gptcorrect} correctly. A majority of these are one-step questions, showing that ChatGPT with minimal prompt tuning can can only reason at very shallow depths.

\begin{table}[ht!]
\begin{center}
\begin{minipage}{\textwidth}
\caption{Questions correctly solved by ChatGPT}\label{tab:gptcorrect}%
\begin{tabular}{@{}c|ll|l@{}}
\toprule
Question & Premise & Target & Comments \\
\midrule
1 & $\neg(\neg p)$ & $p$ & One-step Double Negation \\
3 & $p\to (q\to r)$ & $(p\land q) \to r$ & Uses Modus Ponens cases \\
6 & $\neg p \land \neg q$ & $\neg (p\lor q)$ & One-step De Morgan's law \\
7 & $\neg(p\land \neg q)\lor q$ & $\neg p \lor q$ & De Morgan's $\rightarrow$ Idempotence \\
9 & $(p\lor q)\land (p\lor r)$ & $p\lor (q\land r)$ & One-step Distributivity \\
\botrule
\end{tabular}
\end{minipage}
\end{center}
\end{table}

\end{appendices}

\end{document}